\documentclass[journal]{IEEEtran}

\usepackage{cite}

\usepackage{color}

\usepackage[pdftex]{graphicx}
\graphicspath{{Figs/}{../pdf/}{../jpeg/}}
\DeclareGraphicsExtensions{.pdf,.jpeg,.png}

\usepackage{amsmath}

\usepackage{amssymb}

\usepackage{mathrsfs}

\usepackage{algorithmic}

\usepackage{array}

\usepackage{fixltx2e}

\usepackage{stfloats}

\usepackage{url}

\usepackage{booktabs}

\usepackage{enumerate}

\usepackage{setspace}

\usepackage{amsmath}
\usepackage{amsthm}
\newtheorem{theorem}{Theorem}[section]
\newtheorem{lemma}[theorem]{Lemma}

\newtheorem{proposition}[theorem]{Proposition}

\newtheorem{remark}[]{Remark}
\newtheorem{assumption}[]{Assumption}

\usepackage[normalem]{ulem}
\newcommand\rout{\bgroup\markoverwith{\textcolor{red}{/}}\ULon} % red xout

\usepackage[prependcaption,colorinlistoftodos]{todonotes}
\definecolor{ORANGE}{RGB}{234,100,10}
\definecolor{ORANGE1}{RGB}{150,70,10}
\begin{document}

\title{Impacts of Grid Structure on PLL-Synchronization Stability of Converter-Integrated Power Systems}

\author{Linbin Huang, Huanhai Xin, Wei Dong and Florian D{\"o}rfler
\thanks{L. Huang is with the College of Electrical Engineering at Zhejiang University, Hangzhou, China, and the Department of Information Technology and Electrical Engineering at ETH Zurich, Switzerland. (Email: \text{huanglb@zju.edu.cn})}
\thanks{H. Xin and W. Dong are with the College of Electrical Engineering at Zhejiang University, Hangzhou, China. (Email: \text{xinhh@zju.edu.cn}, \text{eedongwei@zju.edu.cn})}
\thanks{F. D{\"o}rfler is with the Department of Information Technology and Electrical Engineering at ETH Zurich, Switzerland. (Email: \text{dorfler@ethz.ch})}}

%\onecolumn

\maketitle

%\doublespacing

\begin{abstract}
Small-signal instability of grid-connected power converters may arise when the converters use a phase-locked loop (PLL) to synchronize with a weak  grid. Commonly, this stability problem (referred as PLL-synchronization stability in this paper) was studied by employing a single-converter system connected to an infinite bus, which however, omits the impacts of power grid structure and the interactions among multiple converters. Motivated by this, we investigate how the grid structure affects PLL-synchronization stability of multi-converter systems. By using Kron reduction to eliminate the interior nodes, an equivalent reduced network is obtained which  contains only the converter nodes. We explicitly show how the Kron-reduced multi-converter system can be decoupled into its modes. This modal representation allows us to demonstrate that the smallest eigenvalue of the grounded Laplacian matrix of the Kron-reduced network dominates the stability margin. We also carry out a sensitivity analysis of this smallest eigenvalue to explore how a perturbation in the original network affects the stability margin. On this basis, we provide guidelines on how to improve the  PLL-synchronization stability of multi-converter systems by PLL-retuning, proper placement of converters or enhancing some weak connection in the network. Finally, we validate our findings with simulation results based on a 39-bus test system.
\end{abstract}


%\IEEEpeerreviewmaketitle

\section{Introduction}

\IEEEPARstart{W}{ith} the development of renewables, energy storage systems, microgrids, and high-voltage DC (HVDC) systems, more and more power-electronic devices (i.e., power converters) are integrated in modern power systems, and we can foresee a future of converter-dominated power systems \cite{carrasco2006power}. The dynamics of power converters are usually different from synchronous generators (SGs), especially when they are operated in grid-following mode which utilizes a phase-locked loop (PLL) for grid-synchronization \cite{rocabert2012control, FM-FD-GH-DH-GV:18}.

The SG has physical rotating part which determines the angular frequency and make the SG synchronize with the power grid spontaneously \cite{dorfler2012synchronization}, while the converters are composed of static semiconductor components which have high controllability and flexibility. Particularly, the grid-following converter utilizes a  PLL to realize voltage orientation and grid frequency tracking, thereby sharing totally different synchronization mechanism from SGs \cite{blaabjerg2006overview}. The PLL determines the power angle and the largest time constant of the closed-loop converter system and thus dominates its input/output dynamics as seen from the grid side.
Conventionally, the design of a PLL assumes constant voltage magnitude and stiff frequency at the measuring point, so the dynamics of the  PLL are decoupled from the other parts of the converter system (e.g., current control loop, \emph{LCL}, etc.). In this way, the PLL can be regarded as a second-order filter to track the grid frequency \cite{golestan2015conventional}, and a larger bandwidth improves the frequency tracking capability. 

However, in a real converter system, the dynamics of PLL can strongly interact with the other parts, especially when the converter is integrated in weak (non-stiff) grids that feature low short-circuit ratios. Moreover, the PLL induces negative resistor effect on the equivalent input/output admittance model of the converter, which may result in small-signal instability and thus oscillations of the PLL's output \cite{wang2018unified, wen2016analysis}. 
In the rest of the paper, we will refer to this stability issue as PLL-synchronization stability since it is caused by PLL.

The PLL-synchronization instability has been widely analyzed via a single converter connected to an infinite bus, which showed that instabilities may arise under high grid impedance (i.e., weak grid condition) \cite{liu2018frequency, huang2018adaptive}. However, the interactions among multiple converters and the impacts of grid structure cannot be revealed in such a system, thus it is still unclear how the PLL-synchronization instability results from different grid structure of a multi-converter system.

Commonly, the small-signal stability of a multi-converter system (e.g., in microgrids or in low-voltage distribution grids) is evaluated by deriving the state-space model of the entire system and then obtaining the eigenvalues, but this method offers little physical insights into the stability mechanism. A reduced-order model was proposed in \cite{purba2017network} to study multi-converter systems, which uses an aggregate model to represent the dynamics of multiple converters. 
However, the aggregate model omits the dynamic interactions among the converters. In \cite{jafarpour2019small}, the small-signal stability of an inverter network was studied by using time-scale analysis, which leads to an analytic sufficient condition for local exponential stability. In \cite{xin2017generalized}, the converters are modeled as transfer function matrices to describe how voltage perturbations affect the active and reactive power, and the dynamics of the converters are decoupled to understand the overall system stability, which lays the foundation of the analysis in this paper. However, the model in \cite{xin2017generalized} can hardly deal with networks that have interior (non-converter) buses. 
In general, it is not fully understood how interactions between PLLs and other grid components as well as the grid structure give rise to instability.
In short, PLL-synchronization stability in multi-converter systems still remains to be investigated thoroughly.

The grid structure, i.e., the topology and the coupling strength (admittance) of electric transmission network, has been shown to have a significant impact on the stability issues of conventional power systems, e.g., transient and small signal stability \cite{hill2006power,dorfler2012synchronization,dorfler2013synchronization,song2017network}. Moreover, in microgrids that consist of droop-controlled (grid-forming) converters, the grid structure will significantly affect the synchronization (or large-disturbance stability) of the converters \cite{simpson2013synchronization}.
However, for a multi-converter (PLL-based) system, it still remains unknown how the grid structure affects the stability margin. This paper aims at filling this gap, as the PLL-synchronization instability has become a major concern for PLL-based converters \cite{wen2015impedance,huang2018adaptive}. Particularly, we attempt to answer the following questions that motivate this paper. How do the PLL-based converters interact with each other via the transmission network?
How does the grid structure affect the PLL-synchronization stability of the system? What is the effect of a perturbation in the transmission network on the stability? Which transmission line will most sensitively influence the stability margin? 

To provide insightful answers to these questions, we pursue an analytic approach based on a simplified model, though our results are also numerically validated on a detailed simulation model. In particular, we derive a small-signal model for converters to describe the synchronization process achieved by PLL and to obtain the PLL-synchronization stability margin. We provide an explicit analysis of a multi-converter power system under simplifying assumptions on the interconnecting lines and loads so that the model is amenable to a simultaneous diagonalization procedure. As a result, we reveal the device-level and system-level aspects of the overall power grid stability affected by the grid structure. The stability margin is determined by the stability of a single converter connected to infinite bus through a line of strength corresponding to the smallest eigenvalue of the underlying grounded and Kron-reduced network Laplacian matrix. We study the effects of  converter parameters (e.g., the PLL bandwidth) as well as the grid structure on the stability margin by means of explicit eigenvalue sensitivity calculations. We illustrate our insightful results and the utility of our approach through a numerical multi-converter case study and via nonlinear simulations.

Based on our insights, we offer constructive countermeasures to PLL-induced instabilities, such as PLL retuning or, on the planning level, proper placement of converter-induced generation or enforcing weak grid connections.

The rest of this paper is organized as follows: Section II presents modeling of multi-converter systems and Kron reduction of the electrical network. Section III analyzes the PLL-synchronization stability and shows how the stability margin is related to the grid structure. 
Section IV provides sensitivity analysis of the grid structure on the stability. Simulation results are provided in Section V. Section VI concludes the paper.

\section{Modeling of Multi-Converter Systems}

\begin{figure}[!t]
	\centering
	\includegraphics[width=3.2in]{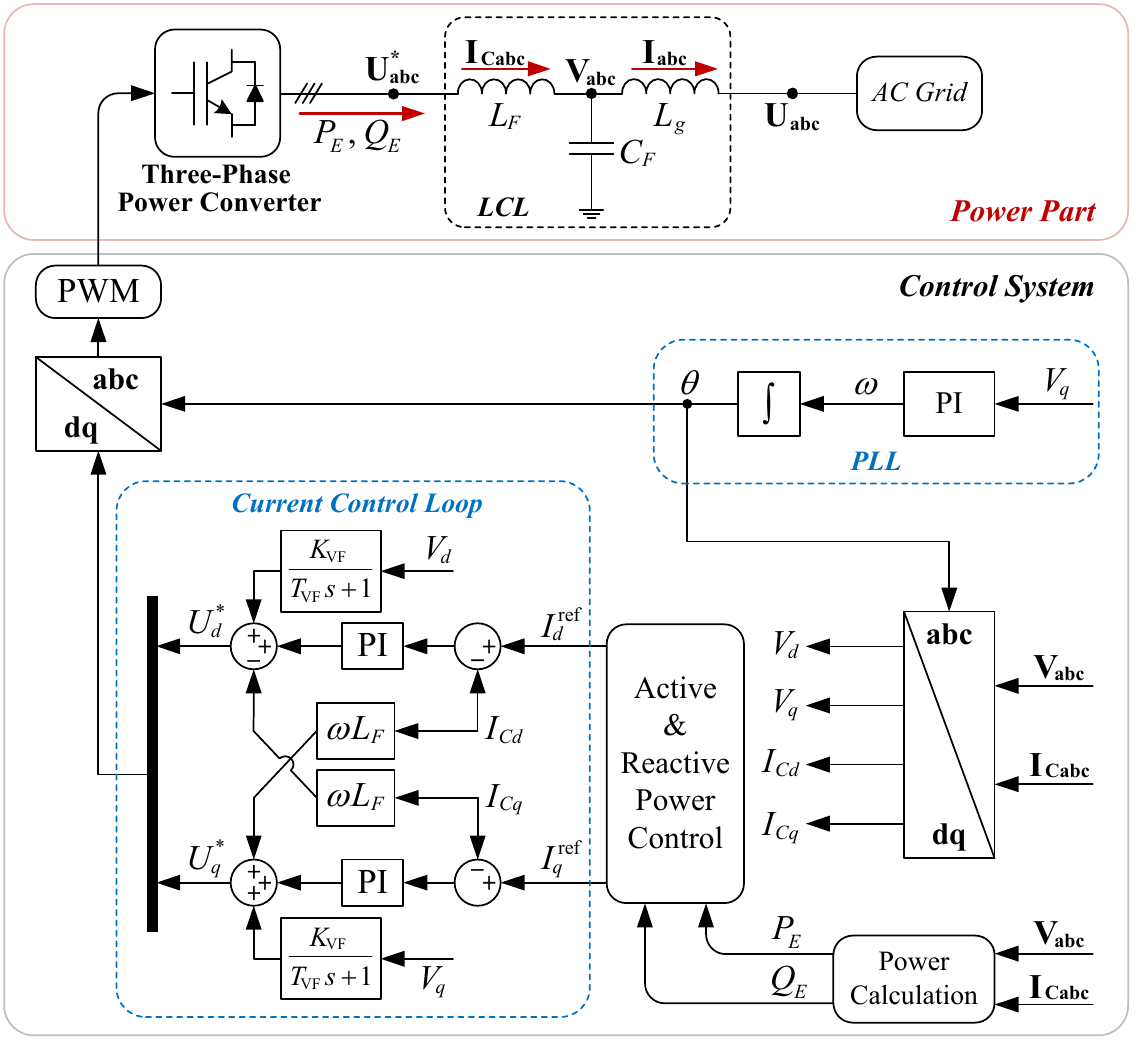}
	\vspace{-3mm}
	%\DeclareGraphicsExtensions.
	\caption{One-line diagram of a PLL-based power converter.}
	\vspace{-3mm}
	\label{Fig_Single_Converter}
\end{figure}

Fig.\ref{Fig_Single_Converter} shows a three-phase power converter which is connected to the ac grid via an \emph{LCL} filter. The converter applies a PLL for grid synchronization. ${{\bf{V}}_{{\bf{abc}}}}$ is the three-phase capacitor voltage of the \emph{LCL}. ${{\bf{I}}_{{\bf{Cabc}}}}$ is the converter-side current. ${{\bf{I}}_{{\bf{abc}}}}$ is the current that injected into the ac grid. ${\bf{U}}_{{\bf{abc}}}^{\bf{*}}$ is the converter's voltage output that determined by the modulation. ${{\bf{U}}_{{\bf{abc}}}}$ is the terminal voltage of the ac grid. When the power grid is balanced (i.e., only the positive-sequence components need to be considered), these three-phase signals can be represented by two-dimensional vectors in the static $\alpha\beta$ frame or the synchronously-rotating $dq$ frame.

\subsection{Admittance Matrix of PLL-based Power Converters}

Consider a three-phase grid-connected power converter. Under balanced grid condition, let $u$ be the two-dimensional voltage vector at the terminal of this converter, and let $i$ be the two-dimensional current vector injected into this converter from the terminal. Assuming that $u$ and $i$ are in the same coordinate (e.g., $\alpha\beta$ frame or $dq$ frame), the linearized model of this converter can be represented by a $2 \times 2$ transfer function matrix ${\bf Y_C(s)}$ (defined as admittance matrix) as $i={\bf Y_C(s)}u$.
Note that in the admittance model we omit the time-delay effect of PWM to simplify the expression, which is reasonable when focusing on the time-scale of PLL dynamics \cite{harnefors2007input}.

We remark that ${\bf Y_C(s)}$ is the closed-loop transfer function of the converter system when its terminal is directly connected to a stiff grid. Hence, ${\bf Y_C(s)}$ should be designed to be stable, i.e., the closed-loop poles solely lie in the open left-half of the complex $s$ plane, such that the converter can operate stably when connected to a stiff grid.

The admittance matrix of the PLL-based converter in Fig.\ref{Fig_Single_Converter} in the global $dq$-frame is 
\begin{equation}
- \left[ {\begin{array}{*{20}{c}}
	{\Delta {I'_{d}}}\\
	{\Delta {I'_{q}}}
	\end{array}} \right] = {{\bf{Y}}_{{\bf{C}}}}(s)\left[ {\begin{array}{*{20}{c}}
	{\Delta {U'_d}}\\
	{\Delta {U'_q}}
	\end{array}} \right]\,,\\	
\label{eq:Y_C}
\end{equation}
where ${\left[ {\begin{array}{*{20}{c}} {\Delta {I'_{d}}}&{\Delta {I'_{q}}} \end{array}} \right]^\top}$ and ${\left[ {\begin{array}{*{20}{c}} {\Delta {U'_d}}&{\Delta {U'_q}} \end{array}} \right]^\top}$ are respectively the perturbed vectors of the converter's current output and terminal voltage in the global $dq$-frame. We provide a detailed derivation of this admittance matrix in Appendix \ref{Appenxix:admittance Matrix} based on complex vectors and transfer function matrices \cite{harnefors2007modeling}.

\subsection{Coupling of Multiple Converters via Electrical Network}

In this subsection, we will show how the converters' dynamics are coupled via the electrical network.
Since we are interested in providing insightful and analytic insights into how the converter and network parameters affect the overall system stability, we make the following   assumptions leading to a simplified -- albeit analytically tractable -- model.

\begin{assumption}
\label{assumptions}
\leavevmode 
\begin{enumerate}[(i)]
	
	\item All the converters adopt the same control scheme and use the same parameters, thereby having the very same equivalent admittance matrices when formulated as (\ref{eq:Y_C});
	
	\item all the lines have the same $R/L$ ratio;
	
	\item the loads are simple constant current loads in the global {dq}-frame that play no role for the linearized model. 
	
\end{enumerate}
\end{assumption}

These above assumptions lead to a final multi-converter model that is amenable to a simultaneous diagonalization procedure and can be decoupled into several subsystems corresponding to the Laplacian modes of the network interaction. This setup includes any low-voltage grid without synchronous generators and connected to an infinite bus (e.g., wind farms and microgrids in grid-connected mode). Moreover, we particularly focus on the network that interconnects the converters.

Fig.\ref{Fig_Multi_Converter} shows a multi-converter system, in which the converters are interconnected via an electrical network (the gray part within the dash line). An infinite bus (with angle $\theta_{G}$) is needed in this scenario because the  PLL-based converters need a frequency reference. For $n,m \in \mathbb{N}$, the electrical network contains $m$ converter nodes (denoted by $1{\rm{st}},2{\rm{nd}},...,m{\rm{th}}$), $n-m$ interior nodes i.e., the gray nodes in Fig.\ref{Fig_Multi_Converter} (denoted by $\left( {m + 1} \right){\rm{th}},...,n{\rm{th}}$), and one infinite-bus node (the $\left( {n + 1} \right){\rm{th}}$ node). 
The transmission lines are assumed to be inductive, and the loads are modeled as constant current sources. For a transmission line that connects node $i$ and node $j$ ($i,j \in \mathcal{I}_{n+1}$, and the set $\mathcal{I}_{n+1} = \left\{{1,...,n+1} \right\}$), the dynamic equation can be expressed in the global coordinate as \cite{huang2019grid,harnefors2007modeling}
\begin{equation}
\begin{split}
\left[ {\begin{array}{*{20}{c}}
	{{\Delta I'_{d,ij}}}\\
	{{\Delta I'_{q,ij}}}
	\end{array}} \right] &= {B_{ij}}{\bf{F}}(s)\left[ {\begin{array}{*{20}{c}}
	{{\Delta U'_{d,i}} - {\Delta U'_{d,j}}}\\
	{{\Delta U'_{q,i}} - {\Delta U'_{q,j}}}
	\end{array}} \right]\,,\\
{\bf{F}}(s) &= \frac{{{1}}}{{{(s+\tau)^2/\omega _0} + \omega _0}}\left[ {\begin{array}{*{20}{c}}
	{s+\tau}&{ {\omega _0}}\\
	{ - {\omega _0}}&{s+\tau}
	\end{array}} \right]\,,		\label{eq:nodeij}
\end{split}
\end{equation}
where ${\left[ {\begin{array}{*{20}{c}} {{\Delta I'_{d,ij}}}&{{\Delta I'_{q,ij}}} \end{array}} \right]^\top}$ is the current vector from node $i$ to node $j$, ${\left[ {\begin{array}{*{20}{c}} {{\Delta U'_{d,i}}}&{{\Delta U'_{q,i}}} \end{array}} \right]^\top}$ is the voltage at node $i$, $B_{ij} = 1/(L_{ij} \times \omega _0)$ is the susceptance between $i$ and $j$, and $\tau$ is the identical $R_{ij}/L_{ij}$ ratio of all the lines.

\begin{figure}[!t]
	\centering
	\includegraphics[width=3in]{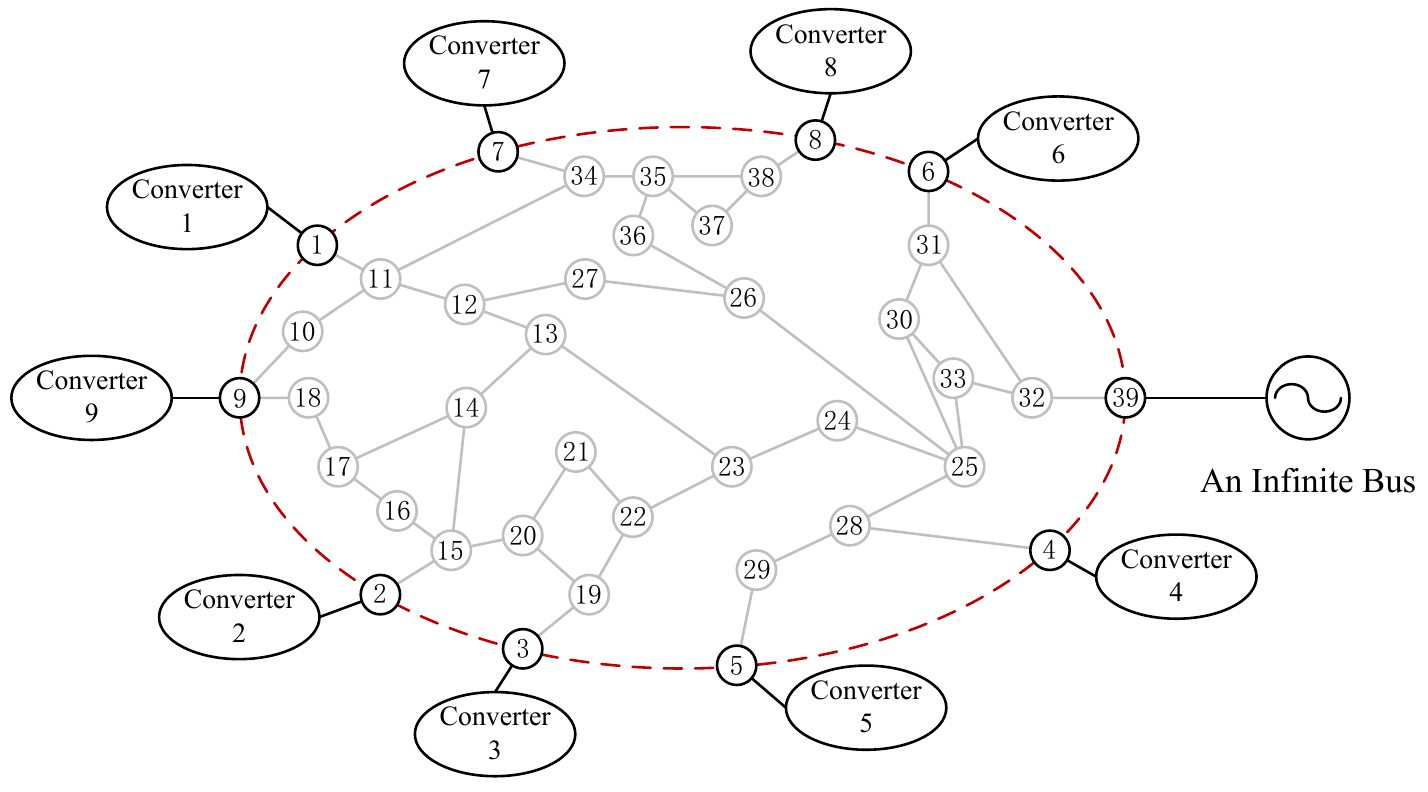}
	\vspace{-3mm}
	%\DeclareGraphicsExtensions.
	\caption{A multi-converter system.}
	\vspace{-3mm}
	\label{Fig_Multi_Converter}
\end{figure}

Since the voltage vector of the infinite bus is constant (i.e., $1+j0$ p.u.) in the global coordinate, the infinite bus can be assigned as the grounded node in the small-signal model. Hence, the electrical network contains self-loops (i.e., edges between the grounded node and the other nodes) \cite{dorfler2013kron}. 

Let $Q \in \mathbb{R}^{n \times n}$ be the grounded Laplacian matrix of a electrical network that encodes the line topology and weightings, calculated by
\begin{equation}
\begin{split}
{Q_{ij}} &=  - {B_{ij}},\;i \ne j\,,\\
{Q_{ii}} &= \sum\limits_{j = 1}^{n} {{B_{ij}}} + B_{i,n+1}\,,		\label{eq:Qmatrix}
\end{split}
\end{equation}
and $Q \otimes {\bf{F}}(s)$ is the corresponding admittance matrix ($\otimes$ denotes the Kronecker product).

Then, by eliminating the interior nodes through Kron reduction one obtains an equivalent grounded network that only contains the converter nodes, as shown in Fig.\ref{Fig_Network_KR}. The grounded Laplacian matrix of the Kron-reduced network can be calculated by
\begin{equation}
{Q_{{\rm{red}}}} = {Q_1} - {Q_2} \times Q_4^{ - 1} \times {Q_3}\,,		\label{eq:Qred}
\end{equation}
where ${Q_1} \in {\mathbb{R}^{m \times m}}$, ${Q_2} \in {\mathbb{R}^{m \times \left( {n - m} \right)}}$, ${Q_3} \in {\mathbb{R}^{\left( {n - m} \right) \times m}}$ and ${Q_4} \in {\mathbb{R}^{\left( {n - m} \right) \times \left( {n - m} \right)}}$ are the submatrices of $Q$ as
\begin{equation}
Q = \left[ {\begin{array}{*{20}{c}}
	{{Q_1}}&\vline& {{Q_2}}\\
	\hline
	{{Q_3}}&\vline& {{Q_4}}
	\end{array}} \right].
\label{eq:subQ}
\end{equation}

\begin{figure}[!t]
	\centering
	\includegraphics[width=3.5in]{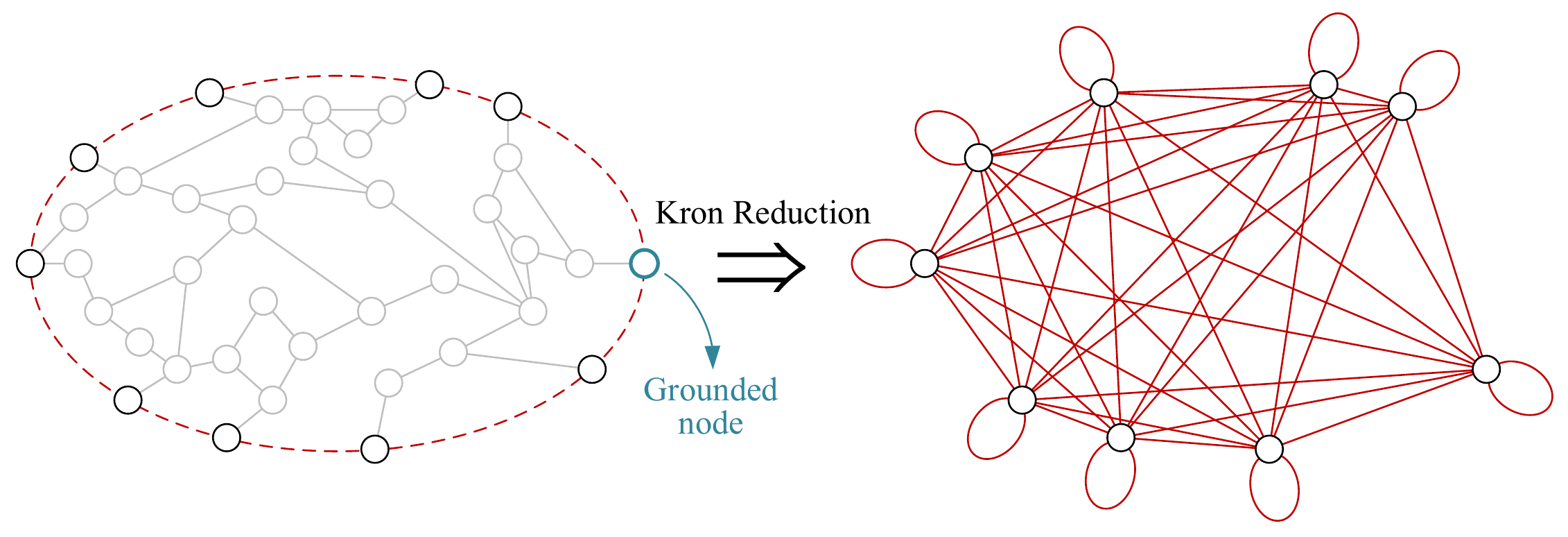}
	\vspace{-8mm}
	%\DeclareGraphicsExtensions.
	\caption{Kron reduction of the multi-converter system.}
	\vspace{-3mm}
	\label{Fig_Network_KR}
\end{figure}

It can be seen from the Kron-reduced network in Fig.\ref{Fig_Network_KR} that the converters interact with each other through the equivalent network, and the self-loops reflect the interactions between the converters and the infinite bus. Combining (\ref{eq:Qred}) and (\ref{eq:nodeij}) yields the network dynamics represented by the admittance matrix
\begin{equation}
\underbrace {\left[ {\begin{array}{*{20}{c}}
	{\Delta {I'_{d,1}}}\\
	{\Delta {I'_{q,1}}}\\
	\vdots \\
	{\Delta {I'_{d,m}}}\\
	{\Delta {I'_{q,m}}}
	\end{array}} \right]}_{\Delta {\bf{I'}}} = {Q_{{\rm{red}}}} \otimes {\bf{F}}(s)\underbrace
{\left[ {\begin{array}{*{20}{c}}
	{\Delta {U'_{d,1}}}\\
	{\Delta {U'_{q,1}}}\\
	\vdots \\
	{\Delta {U'_{d,m}}}\\
	{\Delta {U'_{q,m}}}
	\end{array}} \right]}_{\Delta {\bf{U'}}}		\label{eq:network-side}
\end{equation}
where ${\left[ {\begin{array}{*{20}{c}} {{\Delta I'_{d,i}}}&{{\Delta I'_{q,i}}} \end{array}} \right]^\top}$ is the current injection at node $i$ provided by the $i\rm{th}$ converter.

On the other hand, based on assumption (i), Eq.(\ref{eq:Y_C}) can be extended to represent the dynamics of all the converters as
\begin{equation}
\Delta {\bf{I'}} =  - {I_m} \otimes {{\bf{Y}}_{\bf{C}}}(s) \times \Delta {\bf{U'}}\,,	\label{eq:converter-side}
\end{equation}
where $I_m$ denotes the $m$-dimensional identity matrix.

Then, by combining the converter-side dynamics (i.e., (\ref{eq:converter-side})) and the network-side dynamics (i.e., (\ref{eq:network-side})) one obtains the closed-loop diagram of the multi-converter system that reflects the overall system dynamics, as depicted in Fig.\ref{Fig_Closed_loop}. Note that additive disturbances in the current injections can be conveniently considered in this closed-loop block diagram. Then, the open-loop transfer function matrix can be formulated by
\begin{equation}
{\bf T_{O}}(s) = {I_m} \otimes {\bf Y_C}(s) \times \left[ {Q_{{\rm{red}}}} \otimes {\bf{F}}(s) \right ]^{-1}\,.
\label{eq:openloopT}
\end{equation}
Moreover, the system is stable if and only if ${\rm det}|I_{2m}+{\bf T_{O}}(s)|=0$ is Hurwitz.

\begin{figure}[!t]
	\centering
	\includegraphics[width=1.6in]{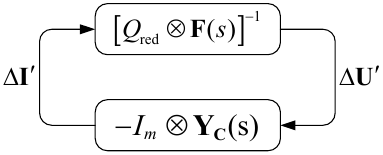}
	\vspace{-3mm}
	%\DeclareGraphicsExtensions.
	\caption{Closed-loop diagram of the multi-converter system.}
	\vspace{-3mm}
	\label{Fig_Closed_loop}
\end{figure}

\subsection{Decoupling of the Multi-Converter System}

The following proposition shows how the multi-converter system in Fig.\ref{Fig_Closed_loop} can be decoupled.

\begin{proposition}[Decoupling of multi-converter system]
\label{proposition: Decoupling of multi-converter system}
The $m$-converter system in \eqref{eq:openloopT} can be decoupled into $m$ subsystems, and the open-loop transfer function matrix of the $k{\rm th}$ subsystem is 
\begin{equation}
{\bf T}_{\bf O}^k(s) = {\bf Y_C}(s) \times \left[ \lambda_k \times {\bf{F}}(s) \right ]^{-1}\,,
\label{eq:openloopTi}
\end{equation}
where ${0 < {\lambda _1} \leq {\lambda _2} \leq ... \leq {\lambda _m}}$ are the eigenvalues of $Q_{{\rm{red}}}$.
Moreover, the system in \eqref{eq:openloopT} is stable if and only if ${\rm det}|I_{2}+{\bf T}_{\bf O}^k(s)|=0$ is Hurwitz for every $k \in \{1,...,m\}$.
\end{proposition}
\begin{proof}
	
	Considering that $Q_{{\rm{red}}}$ contains self-loops (the original network is connected to the infinite bus), there exists an invertible matrix $T$ to diagonalize $Q_{{\rm{red}}}$ as
	\begin{equation}
		{T^{ - 1}}{Q_{{\rm{red}}}}T = \Lambda  = {\rm{diag}}\left\{ {{\lambda _1},{\lambda _2},...,{\lambda _m}} \right\}\,.	\label{eq:diag}
	\end{equation}

	Consider the following coordinate transformation
	\begin{equation}
		\begin{split}
			\Delta {\bf{I'_{T}}} &= ({T^{ - 1}} \otimes {I_2}) \times \Delta {\bf{I'}}\,,\\
			\Delta {\bf{U'_{T}}} &= ({T^{ - 1}} \otimes {I_2}) \times \Delta {\bf{U'}}\,,\\
			%\Delta \tau ' &= ({T^{ - 1}} \otimes {I_2}) \times \Delta \tau\,,
		\end{split}		\label{eq:coordinate_transf}
	\end{equation}
	that makes (\ref{eq:network-side}) and (\ref{eq:converter-side}) become
	\begin{equation}
		\underbrace{\left[ {\begin{array}{*{20}{c}}
					{\Delta {I'_{Td,1}}}\\
					{\Delta {I'_{Tq,1}}}\\
					\vdots \\
					{\Delta {I'_{Td,m}}}\\
					{\Delta {I'_{Tq,m}}}
			\end{array}} \right]}_{\Delta {\bf{I'_{T}}}} = \Lambda \otimes {\bf{F}}(s)\underbrace
		{\left[ {\begin{array}{*{20}{c}}
					{\Delta {U'_{Td,1}}}\\
					{\Delta {U'_{Tq,1}}}\\
					\vdots \\
					{\Delta {U'_{Td,m}}}\\
					{\Delta {U'_{Tq,m}}}
			\end{array}} \right]}_{\Delta {\bf{U'_T}}}		\label{eq:network-side-T}
	\end{equation}
	\begin{equation}
		\Delta {\bf{I'_{T}}} =  - {I_m} \otimes {{\bf{Y}}_{\bf{C}}}(s) \times \Delta {\bf{U'_T}}\,,		\label{eq:converter-side-T}
	\end{equation}
	which concludes the proof.
\end{proof}
	
\begin{figure}[!t]
	\centering
	\includegraphics[width=2.6in]{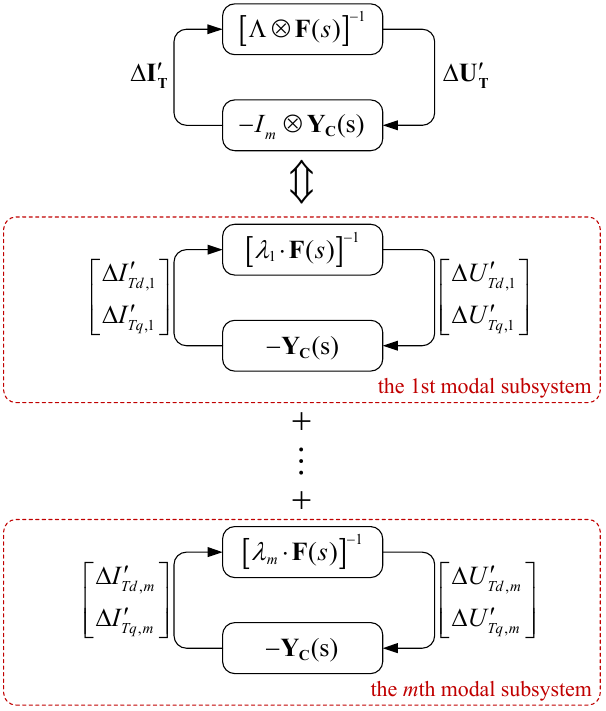}
	\vspace{-3mm}
	%\DeclareGraphicsExtensions.
	\caption{Decoupling of the multi-converter system {into modal subsystems}.}
	\vspace{-3mm}
	\label{Fig_m_subsystem}
\end{figure}

Accordingly, the closed-loop diagram of the multi-converter system (determined by (\ref{eq:network-side-T}) and (\ref{eq:converter-side-T})) is given in Fig.\ref{Fig_m_subsystem}. Note that the closed-loop diagram in Fig.\ref{Fig_m_subsystem} is equivalent to that in Fig.\ref{Fig_Closed_loop} (only with coordinate transformation applied to the inputs and outputs of the system). Hence, they represent the same system dynamics and share the same closed-loop poles.
Since $\Lambda$ and $I_m$ are both diagonal matrices in (\ref{eq:network-side-T}) and (\ref{eq:converter-side-T}), the coordinate transformation in (\ref{eq:coordinate_transf}) intrinsically decouple the multi-converter system into $m$ subsystems.

Note that such a decoupling analysis has become a popular method to study homogeneous network systems \cite{motter2013spontaneous,jiang2017performance,dong2018small}.
In the following, we will refer to these $m$ subsystems as {\em modal subsystems}. 
This decoupling removes the mutual edges of the Kron-reduced network in Fig.\ref{Fig_Network_KR}, and each modal subsystem contains one equivalent self-loop whose weight is $\lambda _k$.

\begin{remark}[Equivalent circuit]
The equivalent circuit of the $k {\rm th}$ modal subsystem is simply a single converter coupled to the infinite bus via the admittance $\lambda_{k} \cdot {\bf{F}}(s)$.
\end{remark}

\begin{remark}[Closed-loop poles of the modal subsystems]
The closed-loop poles of the multi-converter system in \eqref{eq:openloopT} can be obtained by solving ${\rm det}|I_{2m}+{\bf T_{O}}(s)|=0$, or equivalently by solving $\prod\limits_{k=1}^m {\rm det}|I_2+{\bf T}_{\bf O}^k(s)|=0$ (${\bf T}_{\bf O}^k(s)$ is defined in \eqref{eq:openloopTi}), that is, the closed-loop poles of the multi-converter system can be obtained by combining the closed-loop poles of all the modal subsystems.
\end{remark}

\section{PLL-Synchronization Stability Analysis of Modal Subsystems}
\label{Section: PLL-Synchronization Stability Analysis of Modal Subsystems}

The only difference of the decoupled modal subsystems is the admittance between the converter and the infinite bus i.e., $\lambda_{k} \cdot {\bf{F}}(s)$. In the following, we will show how this admittance affects the stability margin of the modal subsystems.

\subsection{Stability Margin Related to Grid Structure}

Fig.\ref{Fig_Sensitivity_F} shows the closed-loop diagram of the $k \rm{th}$ modal subsystem which illustrates how the PLL's output (i.e., $\Delta \theta$) responds to a phase perturbation from the infinite bus (i.e., $\Delta \theta _G$). Note that ${f_{\rm PLL}}(s)$ is defined in \eqref{eq:PLL} of Appendix \ref{Appenxix:admittance Matrix} and for the derivation of ${f_\delta^k }(s)$ we refer to \cite{huang2019grid}, wherein the detailed single-input-single-output model of a single grid-connected power converter was developed.
According to Fig.\ref{Fig_Sensitivity_F}, the open-loop transfer function of the $k{\rm{th}}$ modal subsystem is
\begin{equation}
L_k(s) = {f_\delta^k }(s) \times {f_{\rm PLL}}(s)\,,		\label{eq:L(s)}
\end{equation}
and the sensitivity function of the system is
\begin{equation}
S_k(s) = {1}/{\left[ 1 + L_k(s)\right] }.		\label{S(s)}
\end{equation}

\begin{figure}[!t]
	\centering
	\includegraphics[width=2.5in]{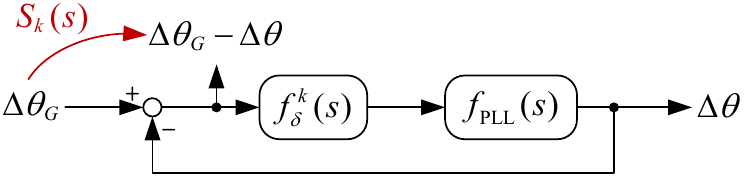}
	\vspace{-3mm}
	%\DeclareGraphicsExtensions.
	\caption{Closed-loop diagram of the $k {\rm th}$ modal subsystem.}
	\vspace{-3mm}
	\label{Fig_Sensitivity_F}
\end{figure}

Since $S_k(s)$ is the transfer function from $\Delta \theta _G$ to $\Delta \theta - \Delta \theta _G$, it represents how the angle disturbance from the infinite bus affects the closed-loop system. We note that \eqref{eq:openloopTi} and \eqref{eq:L(s)} represent the dynamics of the modal subsystems with different choices of inputs and outputs, hence the corresponding closed-loop systems share the same poles. In particular, the signals in the closed-loop diagram in Fig.\ref{Fig_Sensitivity_F} reflect the grid synchronization of the converter via a PLL.
Moreover, as a single-input-single-output system in Fig.\ref{Fig_Sensitivity_F}, Nyquist diagrams can be used to evaluate the system stability: the system is stable if and only if $L_k(s)$ satisfies the Nyquist criterion. Further, the stability margin can be represented by ${1/{\left| {S_k(s)} \right|}_\infty }$ because ${1/{\left| {S_k(s)} \right|}_\infty }$ is intrinsically the shortest distance between point $(-1,0)$ and $L_k(s)$, i.e., Nyquist distance \cite{skogestad2007multivariable}. Note that in our case, $L_k(s)$ is stable as it has no right-half plane pole.

Fig.\ref{Fig_S_lambda} (a) plots the Nyquist diagrams of $L_k(s)$ when different values of the $\lambda _k$ are used (with the system parameters given in Appendix \ref{Appendix: System parameters}), which shows that (regarding the typical parameter set in Appendix \ref{Appendix: System parameters}) the stability margin (i.e., ${1/{| {S_k(s)} |}_\infty }$) is monotonically increasing in $\lambda _k$. 
Further, Fig.\ref{Fig_S_lambda} (b) plots the relationship between $\lambda _k$ and the stability margin, which also shows that the stability margin is monotonically increasing in $\lambda _k$. In this case, the \emph {critical value} of $\lambda _k$ (denoted by $\lambda _C$) is 2.25 (i.e.,  PLL-synchronization instability arises in the  modal subsystem if $\lambda _k < 2.25$).

\begin{figure}[!t]
	\centering
	\includegraphics[width=3in]{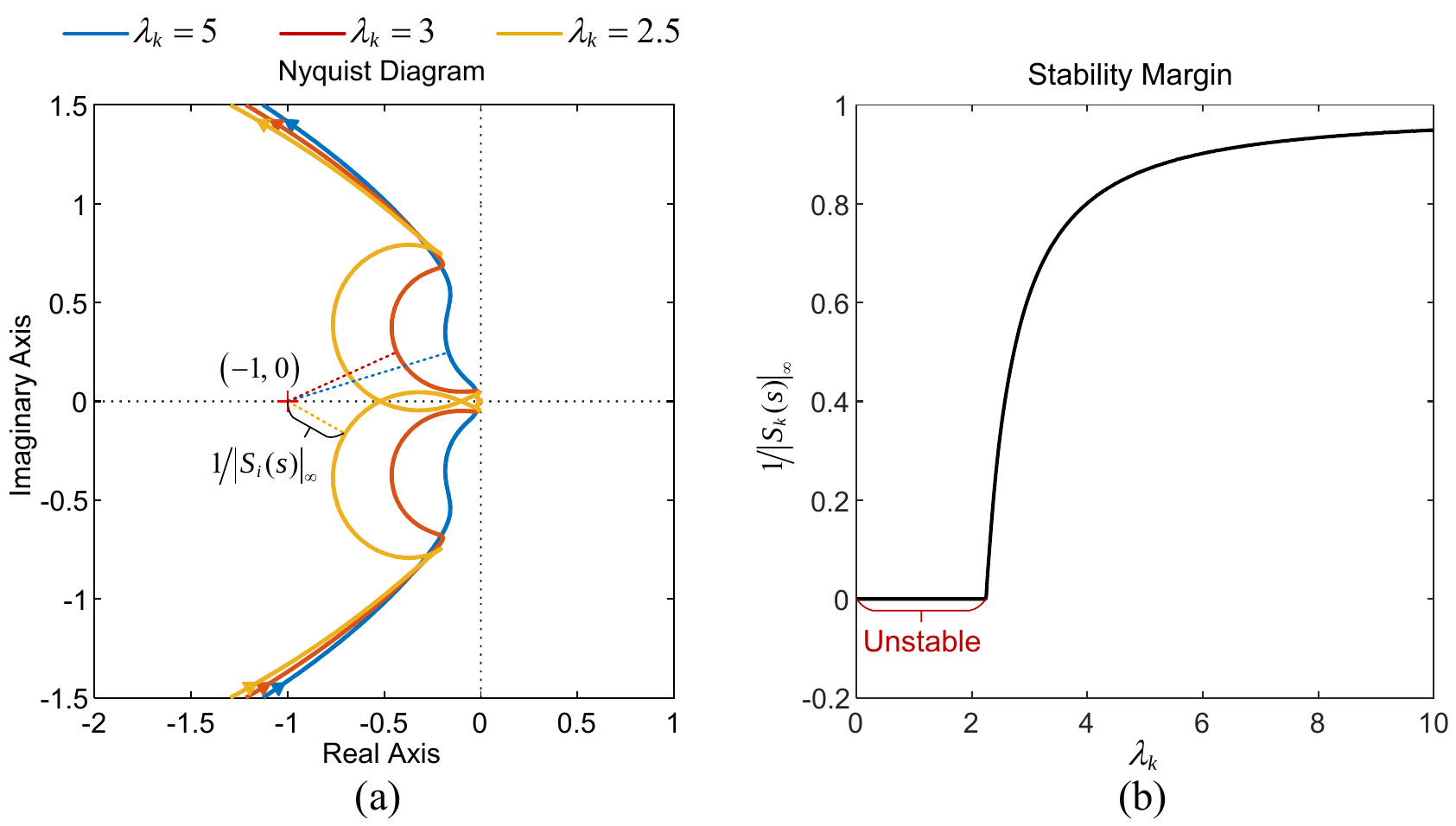}
	\vspace{-3mm}
	%\DeclareGraphicsExtensions.
	\caption{Impacts of ${\lambda _k}$ on stability (a) Nyquist diagram (b) Stability margin.}
	\vspace{-5mm}
	\label{Fig_S_lambda}
\end{figure}

\begin{remark}[Stability margin]
Consider the multi-converter system in \eqref{eq:openloopT} and the typical parameters in Appendix \ref{Appendix: System parameters}.
The PLL-synchronization stability margin of this system equals that of the $1 \rm{st}$ modal subsystem, which is ${1/{\left| {S_1(s)} \right|}_\infty}$. 
\end{remark}

\begin{remark}[Grid structure]
The impacts of different grid structures on the PLL-synchronization stability of the system in \eqref{eq:openloopT} can be interpreted as having different $\lambda_1$ in the $1 {\rm st}$ modal subsystem. If two different grid structure share the same $\lambda_1$, they result in the same stability margin for the system in \eqref{eq:openloopT}.
\end{remark}

This smallest eigenvalue of $Q_{{\rm{red}}}$ can be seen as the connectivity strength of the network \cite{pirani2014spectral, pirani2016smallest, dorfler2018electrical,giordano2016smallest}. For example, \cite{pirani2016smallest} states that for any subset of $k$ nodes of the network (not including the grounded node), $\lambda _1$ has an upper bound which equals to the sum of the edges connecting this subset and the rest of the network (including the grounded node) divided by $k$. Furthermore, $\lambda_{1}$ is a monotone function of the network connectivity, i.e., $\lambda_{1}$ can only increase when the network becomes denser \cite[Proposition 2]{giordano2016smallest}.
In \cite{dong2018small} and \cite{zhang2017generalized}, the smallest eigenvalue of the extended admittance matrix of a multi-infeed system is defined as generalized short circuit ratio to evaluate the power grid strength.

\subsection{Impacts of Converter Parameters}

The previous subsection demonstrates that for a typical set of converter parameters (given in Appendix \ref{Appendix: System parameters}), the PLL-synchronization stability margin is reduced with the decrease of {$\lambda _1$}. Moreover, the system will become unstable if $\lambda _1$ drops below the critical value $\lambda _C$ ($\lambda _C = 2.25$ with the parameters in Appendix \ref{Appendix: System parameters}). In the following, we further explore how the converter parameters affect the stability margin and $\lambda _C$, with a particular focus on the PLL bandwidth.

To begin with, we recall \cite{golestan2015conventional} and \cite{huang2018damping} 
that the relationship between the PLL bandwidth and the PI parameters as follows. By ignoring the coupling between PLL and the other parts of the system, the closed-loop transfer function in Fig.\ref{Fig_Sensitivity_F} can be approximated as
\begin{equation}
\frac{{\Delta \theta }}{{\Delta {\theta _G}}} \approx G_{\rm PLL}(s) = \frac{{s{K_{\rm PLLP}} + {K_{\rm PLLI}}}}{{{s^2} + s{K_{\rm PLLP}} + {K_{\rm PLLI}}}}\,,	\label{eq:PLL_tracking}
\end{equation}
which reflects the tracking capability of the PLL. As a second-order system, the damping ratio of $G_{\rm PLL}(s)$ can be calculated by $\zeta  = {K_{\rm PLLP}}/(2\sqrt {{K_{\rm PLLI}}} )$. By choosing $\zeta  = 1/\sqrt 2 $ for optimal performance, there holds
\begin{equation}
{K_{\rm PLLP}} = \sqrt {2{K_{\rm PLLI}}}\,.		\label{eq:KPLLP_KPLLI}
\end{equation}

With (\ref{eq:KPLLP_KPLLI}), the relationship between the PLL bandwidth $\omega _{\rm BW}$ and $K_{\rm PLLI}$ can be derived as
\begin{equation}
{K_{\rm PLLI}} = \omega _{\rm BW}^2/(2 + \sqrt 5 )\,,
\end{equation}
which meets
\begin{equation*}
{\rm{20}}\lg \left| {{G_{\rm PLL}}(j{\omega _{\rm BW}})} \right| =  - 3{\rm{dB}}\,.
\end{equation*}

Fig.\ref{Fig_S_BW_lambda} shows the stability margin with varying values of the PLL bandwidth and {$\lambda _1$}. Obviously, the stability margin will diminish with the increase of $\omega _{\rm BW}$ but with the decrease of {$\lambda _1$}, resulting in an unstable area. With $\omega _{\rm BW}$ changing from $50 \rm{rad/s}$ to $150 \rm{rad/s}$, $\lambda _C$ is increased from 2.25 to about 2.9, which indicates that the multi-converter system is more prone to instability. For example, if $\lambda _1 = 2.5$ for a certain network, the multi-converter system will be stable by setting $\omega _{\rm BW} = 50 \rm{rad/s}$, while it becomes unstable with $\omega _{\rm BW} = 150 \rm{rad/s}$.

\begin{figure}[!t]
	\centering
	\includegraphics[width=2.5in]{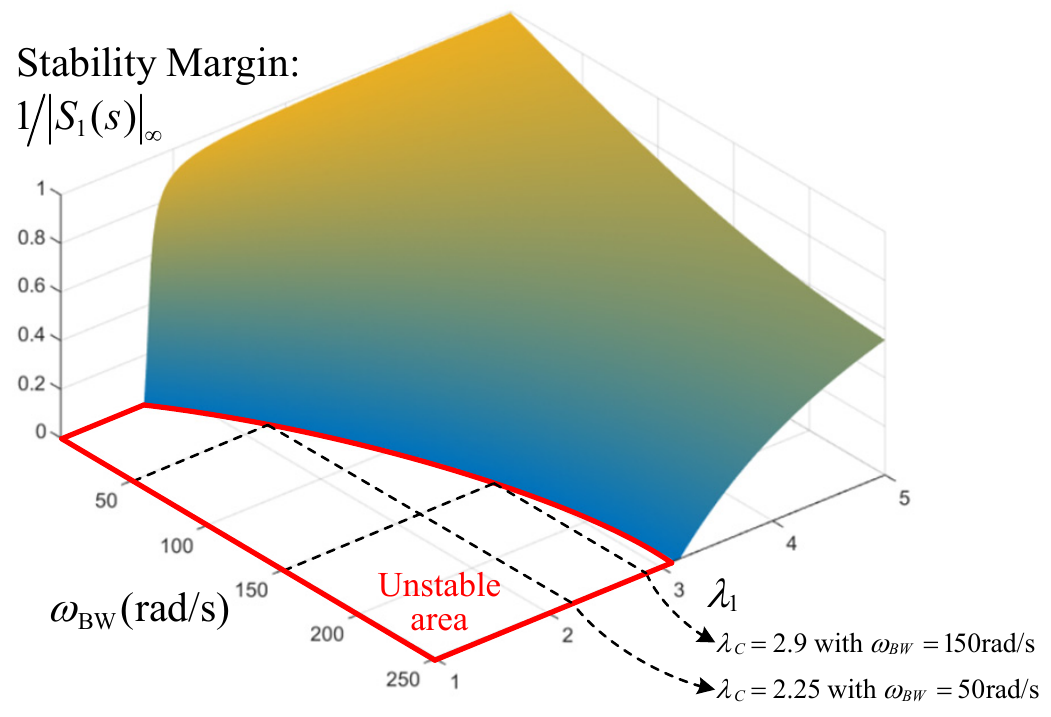}
	\vspace{-3mm}
	%\DeclareGraphicsExtensions.
	\caption{Stability margin affected by PLL bandwidth $\omega _{BW}$ and $\lambda _1$.}
	\vspace{-3mm}
	\label{Fig_S_BW_lambda}
\end{figure}

It is therefore concluded that a) $\lambda _C$ is determined by the converter dynamics while $\lambda _1$ reflects the network characteristic; b) the multi-converter system is stable if and only if $\lambda _C < \lambda _1$; c) increasing $\lambda _C$ (e.g., with higher PLL bandwidth) will make the multi-converter system more prone to {PLL-synchronization} instability; d) increasing $\lambda _1$ (by changing the network structure) ensures a larger stability margin.

As a constructive solution, the PLL can be retuned: notice from Fig.\ref{Fig_S_BW_lambda} that the PLL bandwidth should be accordingly reduced with the decrease of $\lambda_1$ to ensure the stability of multi-converter systems. However, a too-low PLL bandwidth may deteriorate the tracking performance and result in unacceptable overshot. Hence, it is also of great significance to make sure that the grid is strong enough (i.e., with sufficiently large $\lambda_{1}$), which can be done by appropriate grid planning.

\section{Impacts of Grid Structure on PLL-Synchronization Stability}

According to the above analysis, the sensitivity of the stability margin of the multi-converter system in \eqref{eq:openloopT} (denoted by $\mathscr S$) to a perturbation in $B_{ij}$ from the original network (before Kron reduction) can be calculated by
\begin{equation}
\frac{\partial {\mathscr S}}{\partial B_{ij}} = \frac{\partial \left( {1/{\left| {S_1(s)} \right|}_\infty}\right) }{\partial B_{ij}} = \frac{\partial \left( {1/{\left| {S_1(s)} \right|}_\infty}\right) }{\partial \lambda_1} \times \frac{\partial \lambda_1}{\partial B_{ij}}\,,
\end{equation}
where $\frac{\partial \left( {1/{\left| {S_1(s)} \right|}_\infty}\right) }{\partial \lambda_1}$ can be conveniently obtained from Fig.\ref{Fig_S_lambda} (b) or Fig.\ref{Fig_S_BW_lambda}.

In the following, we focus on how $B_{ij}$ affects $\lambda_1$ and thereby the system stability margin. Further, we investigate how to effectively enhance the network and improve the PLL-synchronization stability of the system.

\subsection{Sensitivity of Perturbations in Interior Network}
We now present how a perturbation between interior node $i$ and interior node $j$ in the original network (before Kron reduction) (i.e., $i,j \in \{m+1,...,n\}$) affects the smallest eigenvalue of $Q_{{\rm{red}}}$ ($\lambda_1$).

\begin{lemma}
\label{Lemma:sen_interior_nodes}
Consider a grounded network $Q$ and its Kron-reduced network $Q_{\rm red}$ as defined in \eqref{eq:Qmatrix}-\eqref{eq:subQ}. The sensitivity of $\lambda_1$ to a perturbation in $B_{ij}$ ($i,j \in \{m+1,...,n\}$ can be calculated by
\begin{equation}
\frac{{\partial {\lambda _1}}}{{\partial {B_{ij}}}} = v_1^\top \frac{{\partial {Q_{{\rm{red}}}}}}{{\partial {B_{ij}}}}{u_1}\,,
\label{eq:sen_lambda_k}
\end{equation}
where 
\begin{equation}
\frac{{\partial {Q_{{\rm{red}}}}}}{{\partial {B_{ij}}}} = {Q_{{\rm{ac}}}}\left( {e_{i-m}^{n - m} - e_{j-m}^{n - m}} \right){\left( {e_{i-m}^{n - m} - e_{j-m}^{n - m}} \right)^\top}Q_{{\rm{ac}}}^\top\,,
\label{eq:QredBij}
\end{equation}
$v_1$ is the left eigenvector which meets $v_1^\top {Q_{{\rm{red}}}} = v_1^\top {\lambda _1}$, $u_1$ is the right eigenvector which meets ${Q_{{\rm{red}}}}{u_1} = {\lambda _1}{u_1}$, $e_i^n$ denotes a $n \times 1$ vector with entry 1 at position $i$ and 0 at {all} other positions, and ${Q_{{\rm{ac}}}} =  - {Q_2}Q_4^{ - 1}$ is the accompanying matrix.
\end{lemma}

\begin{proof}
	Recall the eigenvalue sensitivity analysis of state matrix \cite{ma2006eigenvalue, kundur1994power} which leads to \eqref{eq:sen_lambda_k}.
	
	When the susceptance $B_{ij}$ is changed to $B_{ij} + \tau $, $\tau \in \mathbb{R}$, the new grounded Laplacian matrix of the power network becomes 
	\begin{equation}
	\tilde Q = Q + \tau  \times \left( {{e_i^n} - {e_j^n}} \right) \times {\left( {{e_i^n} - {e_j^n}} \right)^\top}\,.
	\end{equation}
	
	According to \cite{dorfler2013kron}, the grounded Laplacian matrix of the Kron-reduced network becomes 
	\begin{equation}
	{\tilde Q_{{\rm{red}}}} = {Q_{{\rm{red}}}} + \frac{{{Q_{{\rm{ac}}}}\left( {e_{i-m}^{n - m} - e_{j-m}^{n - m}} \right)\tau {{\left( {e_{i-m}^{n - m} - e_{j-m}^{n - m}} \right)}^\top}Q_{{\rm{ac}}}^\top}}{{1 + \tau {R_{{\mathop{\rm int}} }}\left[ {i,\;j} \right]}}\,,
	\end{equation}
	where
	\begin{equation*}
	{R_{{\mathop{\rm int}} }}\left[ {i,\;j} \right] \buildrel \Delta \over = {\left( {e_{i-m}^{n - m} - e_{j-m}^{n - m}} \right)^\top}Q_4^{ - 1}\left( {e_{i-m}^{n - m} - e_{j-m}^{n - m}} \right)
	\end{equation*}
	{is the effective resistance (in the interior network) between nodes $i$ and $j$.}
	
	Then, the sensitivity of $Q_{{\rm{red}}}$ to a perturbation in $B_{ij}$ is
	\begin{equation}
	\begin{split}
	\frac{{\partial {Q_{{\rm{red}}}}}}{{\partial {B_{ij}}}} &= \mathop {\lim }\limits_{\tau  \to 0} \frac{{{{\tilde Q}_{{\rm{red}}}} - {Q_{{\rm{red}}}}}}{\tau }\\
	&= {Q_{{\rm{ac}}}}\left( {e_{i-m}^{n - m} - e_{j-m}^{n - m}} \right){\left( {e_{i-m}^{n - m} - e_{j-m}^{n - m}} \right)^\top}Q_{{\rm{ac}}}^\top\,,	
	\label{eq:Qred_Bij}
	\end{split}
	\end{equation}
	where corresponds to the claim in \eqref{eq:QredBij}.
\end{proof}

We remark that the result in Lemma \ref{Lemma:sen_interior_nodes} can be extended to consider a perturbation between an interior node $i$ and the infinite bus (i.e., the self-loop of the interior node $i$), in which case the sensitivity of $Q_{{\rm{red}}}$ becomes
\begin{equation}
\begin{split}
\frac{{\partial {Q_{{\rm{red}}}}}}{{\partial {B_{ij}}}} &= \mathop {\lim }\limits_{\tau  \to 0} \frac{{{{\tilde Q}_{{\rm{red}}}} - {Q_{{\rm{red}}}}}}{\tau }\\
&= {Q_{{\rm{ac}}}}\left( {e_{i-m}^{n - m}} \right){\left( {e_{i-m}^{n - m}} \right)^\top}Q_{{\rm{ac}}}^\top\,.		\label{eq:Qred_Bij_IB}
\end{split}
\end{equation}
By combining \eqref{eq:Qred_Bij_IB} and \eqref{eq:sen_lambda_k} one obtains the sensitivity of $\lambda_1$ to such a perturbation.

\subsection{Sensitivity of Perturbations Among Converter Nodes}

In the following, we discuss the sensitivity of $\lambda_1$ to a perturbation between two converter nodes (or the self-loop of a converter node) in the original network.

\begin{lemma}
\label{Lemma:sen_converter_nodes}
Consider a grounded network $Q$ and its Kron-reduced network $Q_{\rm red}$ as defined in \eqref{eq:Qmatrix}-\eqref{eq:subQ}. The sensitivity of $\lambda_1$ to a perturbation in $B_{ij}$ ($i \in \{1,...,m\}$ and $j \in \{1,...,m\} \cup \{n+1\}$) can be calculated by
\begin{equation}
\frac{{\partial {\lambda _1}}}{{\partial {B_{ij}}}} = v_1^\top W'_{\{i,j\}}{u_1}\,,		\label{eq:lambda_Bij_W'}
\end{equation}
where $W'_{\{i,j\}} = \left( {e_i^m - e_j^m} \right) \times {\left( {e_i^m - e_j^m} \right)^\top}$ if $i,j \in \left\{{1,...,m} \right\}$, while $W'_{\{i,j\}} = e_i^m \times {\left( {e_i^m} \right)^\top}$ if $i \in \{1,...,m\}$ {and} $ j = n+1$.
\end{lemma}

\begin{proof}
	Consider a change of the susceptance between converter node $i$ and converter node $j$ from $B_{ij}$ to $B_{ij} + \tau$, or a change of the susceptance between converter node $i$ and the infinite bus (i.e., node $j = n+1$) from $B_{ij}$ to $B_{ij} + \tau$, which makes the grounded Laplacian matrix of the network become
	\begin{equation}
	\tilde Q = Q + \tau  \times W_{\{i,j\}}\,.
	\end{equation}
	
	Then, according to \cite{dorfler2013kron}, the grounded Laplacian matrix of the Kron-reduced network can be expressed as
	\begin{equation}
	\tilde Q_{\rm{red}} = Q_{\rm{red}} + \tau  \times W'_{\{i,j\}}\,.
	\end{equation}
	
	Further, the sensitivity of $Q_{\rm{red}}$ to a perturbation in $B_{ij}$ can be calculated by
	\begin{equation}
	\frac{{\partial {Q_{{\rm{red}}}}}}{{\partial {B_{ij}}}} = \mathop {\lim }\limits_{\tau  \to 0} \frac{{{{\hat Q}_{{\rm{red}}}} - {Q_{{\rm{red}}}}}}{\tau } = W'_{\{i,j\}}\,.		\label{eq:W'}
	\end{equation}
	
	By combining (\ref{eq:sen_lambda_k}) and (\ref{eq:W'}), the sensitivity of $\lambda _1$ to a perturbation in $B_{ij}$ can be obtained as \eqref{eq:lambda_Bij_W'}.
\end{proof}

We remark that due to the simultaneous diagonalization for decoupling the system, every modal subsystem contains the dynamics of every converter, and the participation factor of the $i\rm{th}$ converter on the $1\rm{st}$ modal subsystem can be calculated by $p_{1i}=v_1^T [e_i^m (e_i^m)^T] {u_1}$ (the proof is similar to that of the participation factor analysis of state matrix \cite{kundur1994power}), which is surprisingly consistent with (\ref{eq:lambda_Bij_W'}) with $j = n+1$. That is, how the converters participate in the $1\rm{st}$ modal subsystems is determined by the grid structure, or to be more specific, by the sensitivity of $\lambda _1$ to perturbations on the self-loops of the converter nodes.

The participation factor $p_{1i}$ also provides insights into the placement of converter-interfaced generation, see e.g. \cite{poolla2017optimal,borsche2015effects}. For example, the converter with the largest participation factor on the $1\rm{st}$ modal subsystem can be moved to other sites that are closer to the converter with the smallest participation factor in order to increase the network connectivity.

\subsection{Case Studies of the 39-bus Test System}
Based on the previous analysis, this subsection provides case studies to illustrate how to improve the system PLL-synchronization stability by enhancing the grid structure. As demonstrated in the previous section, the smallest eigenvalues of $Q_{{\rm{red}}}$ (i.e., $\lambda _1$) determines the small-signal stability margin.  For the multi-converter system in Fig.\ref{Fig_Multi_Converter} (with the parameters given in Appendix \ref{Appendix: System parameters}), the eigenvalues of $Q_{\rm{red}}$ are given in Table \ref{lambda_Qred}. It can be deduced that the multi-converter system is stable because $\lambda _1 = 3.3118 > \lambda _C$; see Section \ref{Section: PLL-Synchronization Stability Analysis of Modal Subsystems}.

\renewcommand\arraystretch{1.4}
\begin{table}
	\centering
	\scriptsize
	\caption{Eigenvalues of $Q_{\rm{red}}$}
	\begin{tabular}{|lll|}
		\hline
		$\lambda _1 = 3.3118$		&	$\lambda _2 = 21.2484$		&	$\lambda _3 = 25.0226$				\\
		$\lambda _4 = 36.0841$		&	$\lambda _5 = 51.3565$		&	$\lambda _6 = 53.7490$				\\
		$\lambda _7 = 61.6484$		&	$\lambda _8 = 70.9915$		&	$\lambda _9 = 77.3948$				\\
		\hline
	\end{tabular}		\label{lambda_Qred}
\end{table}

Fig.\ref{Fig_Sensitivity_Matrix} shows the submatrix of the sensitivity matrix of $\lambda _1$ (denoted by $M$) {with elements} $M_{ij} = \partial \lambda _1 / \partial B_{ij}$, $i,j \in \{m+1,...,n\}$ (calculated by (\ref{eq:QredBij}) and (\ref{eq:sen_lambda_k})), and $M_{ij}$ is set to be $0$ if $i=j$. It can be seen from Fig.\ref{Fig_Sensitivity_Matrix} that $M_{ij} > 0$, which indicates that increasing the susceptance between two interior nodes always helps increase $\lambda _1$ and thus improves the stability of the system. However, the increase of susceptance usually comes at the cost of more investment, e.g., using double-circuit lines instead of single-circuit lines. Hence, it is significant to find the most ``sensitive'' edge that can effectively increase $\lambda _1$.

\begin{figure}[!t]
	\centering
	\includegraphics[width=2.1in]{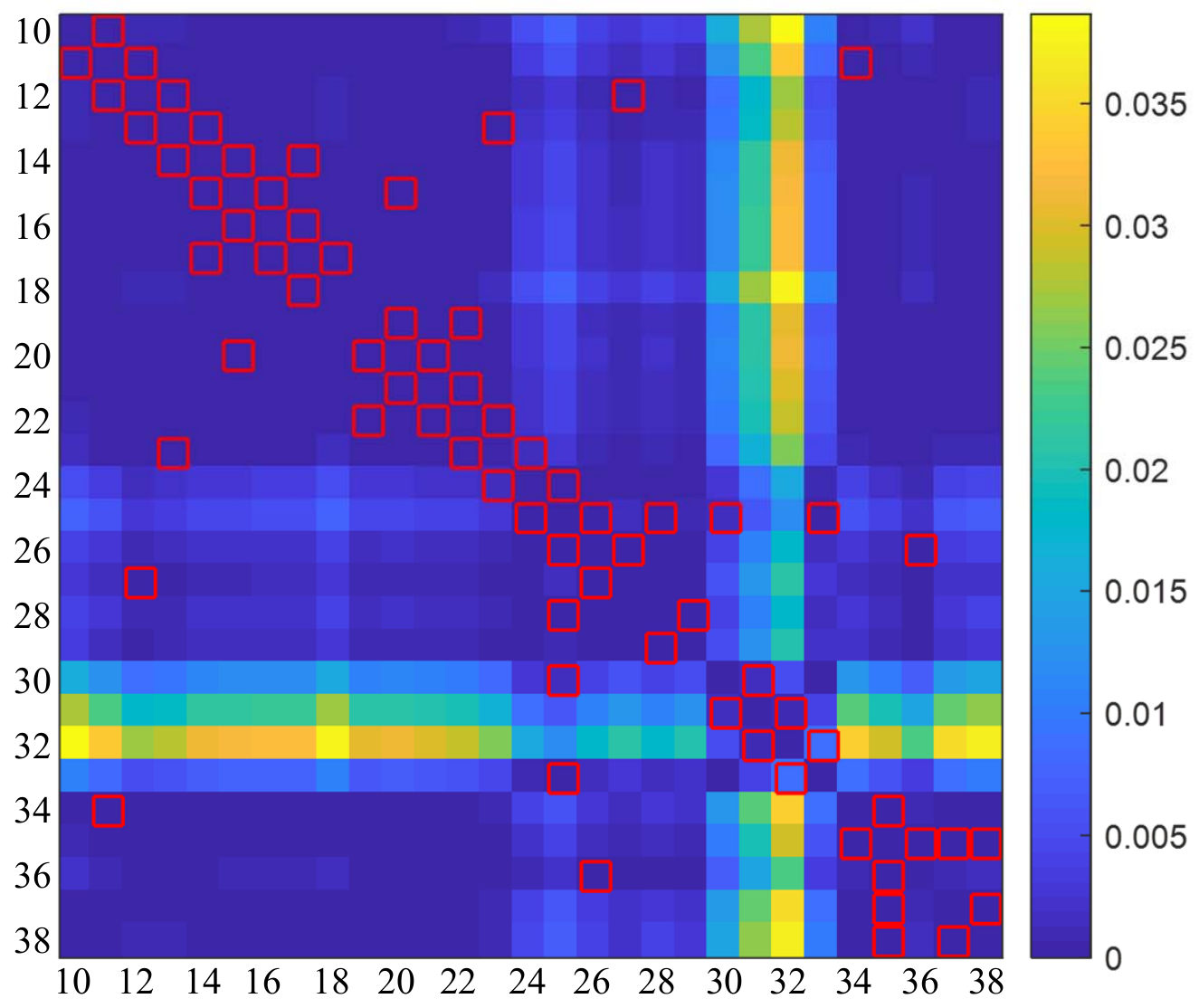}
	\vspace{-3mm}
	%\DeclareGraphicsExtensions.
	\caption{Sensitivity matrix of ${\lambda _1}$ to perturbations in the interior network ({\color{red}{$\square$}} denotes the existing edges in the network).}
	\vspace{-3mm}
	\label{Fig_Sensitivity_Matrix}
\end{figure}

Fig.\ref{Fig_Sensitivity_Matrix} indicates that the stability margin is most sensitive to perturbations affecting links to nodes 10, 18 and 30-33, which are the least connected nodes in Fig.~\ref{Fig_Multi_Converter} as well as the sole connection of the grid to the infinite bus. The maximal entry in Fig.\ref{Fig_Sensitivity_Matrix} is $M_{10,32} = 0.0387$, so increasing the susceptance between node 10 and node 32 has the most significant effect on the stability. On the other hand, as seen from Fig.\ref{Fig_Multi_Converter} that there doesn't exist an edge between node 10 and node 32, increasing the susceptance between these two nodes means building a new transmission line, whose feasibility needs to be further evaluated regarding the economics.

One alternative is to find the most ``sensitive'' edge that has existed in the network.  Fig.\ref{Fig_Sensitivity_Matrix} labels the existing edges with {\color{red}{$\square$}}, and the maximal entry among these edges is $M_{32,33} = 0.0087$, which indicates that increasing the susceptance between node 32 and node 33 can effectively improve the system stability, e.g., adding an additional transmission line. Of course, factors such as economics, reliability and geography should also be taken into account before modifying the grid structure.

Fig.\ref{Fig_Sen_Curves} plots $\lambda _1$ as functions of $B_{32,33}$ and $B_{17,18}$. It shows that with the increase of $B_{32,33}$ from 47.62 to 95.24, $\lambda _1$ is increased from 3.3118 to 3.6014. By comparison, $\lambda _1$ nearly remains the same (from 3.3118 to 3.3172) with the increase of $B_{17,18}$ from 45.91 to 91.82. These results are fully consistent with Fig.\ref{Fig_Sensitivity_Matrix} which illustrates that $\partial \lambda _1 / \partial B_{32,33}$ is much greater than $\partial \lambda _1 / \partial B_{17,18}$.

\begin{figure}[!t]
	\centering
	\includegraphics[width=3.5in]{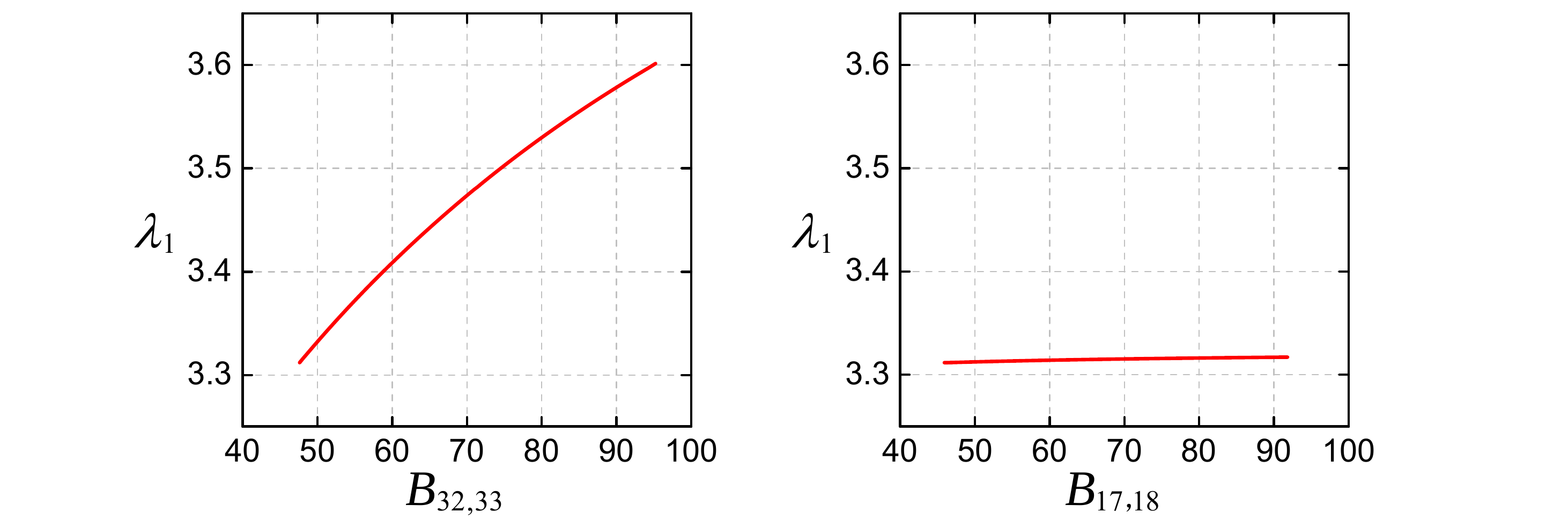}
	\vspace{-7mm}
	%\DeclareGraphicsExtensions.
	\caption{$\lambda _1$ as functions of $B_{32,33}$ and $B_{17,18}$.}
	\label{Fig_Sen_Curves}
\end{figure}

There is a single unique interior node (i.e., node 32) which is directly connected to the infinite bus in the studied network (see Fig.\ref{Fig_Multi_Converter}), and the sensitivity of $\lambda _1$ to a perturbation in $B_{32,39}$ is $\partial \lambda _1 / \partial B_{32,39} = 0.0257$ (calculated by (\ref{eq:Qred_Bij_IB}) and (\ref{eq:sen_lambda_k})). For further illustration, Fig.\ref{Fig_Sen_Curves_IB} plots $\lambda _1$ as a function of $B_{32,39}$, which shows that with the increase of $B_{32,39}$ from 61.27 to 122.54 (e.g., adding a new transmission line), $\lambda _1$ is increased from 3.3118 to 4.3311, thereby improving the PLL-synchronization stability. On the other hand, the stability can be deteriorated with the decrease of $B_{32,39}$. For example, $\lambda _1$ drops below 2.25 and the system becomes unstable when $B_{32,39}$ is less that 30.95, which can be caused by an outage of the transmission line between node 32 and the infinite bus.

\begin{figure}[!t]
	\centering
	\includegraphics[width=1.3in]{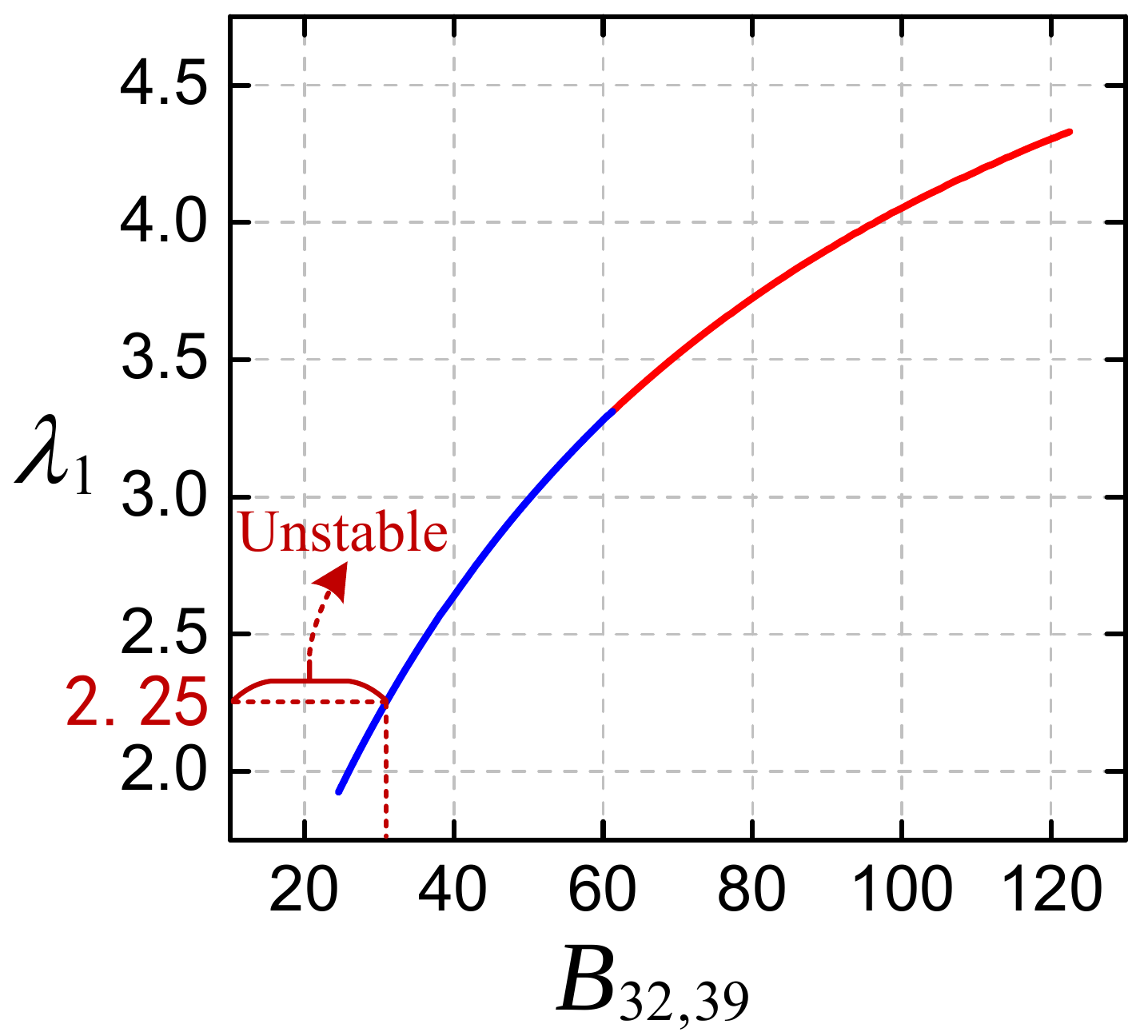}
	\vspace{-3mm}
	%\DeclareGraphicsExtensions.
	\caption{$\lambda _1$ as functions of $B_{32,39}$.}
	\vspace{-3mm}
	\label{Fig_Sen_Curves_IB}
\end{figure}

Table \ref{Participation_Factor} shows the sensitivity of $\lambda _1$ to perturbations on the self-loops of the converter nodes (i.e., the edges between the converter node and the infinite bus), which can be calculated by (\ref{eq:lambda_Bij_W'}). Also, it represents the participation factors of the converters in the $1\rm{st}$ modal subsystem, as discussed in the previous subsection. It can be seen that the $4\rm{th}$, $5\rm{th}$ and $6\rm{th}$ converters (which are closest to the infinite bus in Fig.~\ref{Fig_Multi_Converter}) have the lowest participation, and the remaining converter node have similar and significantly larger participation.

\renewcommand\arraystretch{1.8}
\begin{table}
	\centering
	\scriptsize
	\caption{Sensitivity of ${\lambda _1}$ to perturbations on self-loops.}
	\begin{tabular}{|lll|}
		\hline
		$\frac{\partial \lambda _1} {\partial B_{1,39}} = 0.1269$		&	$\frac{\partial \lambda _1} {\partial B_{2,39}} = 0.1270$		&	$\frac{\partial \lambda _1} {\partial B_{3,39}} = 0.1214$				\\
		$\frac{\partial \lambda _1} {\partial B_{4,39}} = 0.0908$		&	$\frac{\partial \lambda _1} {\partial B_{5,39}} = 0.0978$		&	$\frac{\partial \lambda _1} {\partial B_{6,39}} = 0.0387$				\\
		$\frac{\partial \lambda _1} {\partial B_{7,39}} = 0.1313$		&	$\frac{\partial \lambda _1} {\partial B_{8,39}} = 0.1329$		&	$\frac{\partial \lambda _1} {\partial B_{9,39}} = 0.1332$				\\
		\hline
	\end{tabular}		\label{Participation_Factor}
\end{table}

Fig.\ref{Fig_Sen_Curves_Self-loops} plots $\lambda _1$ as functions of $B_{1,39}$ and $B_{4,39}$, which demonstrates that increasing the susceptance between the converter nodes and the infinite bus has significant effects on improving the stability. $\lambda _1$ is increased from 3.3118 to 6.6073 when $B_{1,39}$ varies from 0 to 50, and it is increased from 3.3118 to 5.3073 when $B_{4,39}$ varies from 0 to 50. Note that in practice, the attainable susceptance between a converter node and the infinite bus is also related to the geographical distance. For example, $B_{1,39} = 50$ may not be attainable if the converter is quite distant from the infinite bus.

\begin{figure}[!t]
	\scriptsize
	\centering
	\includegraphics[width=3.5in]{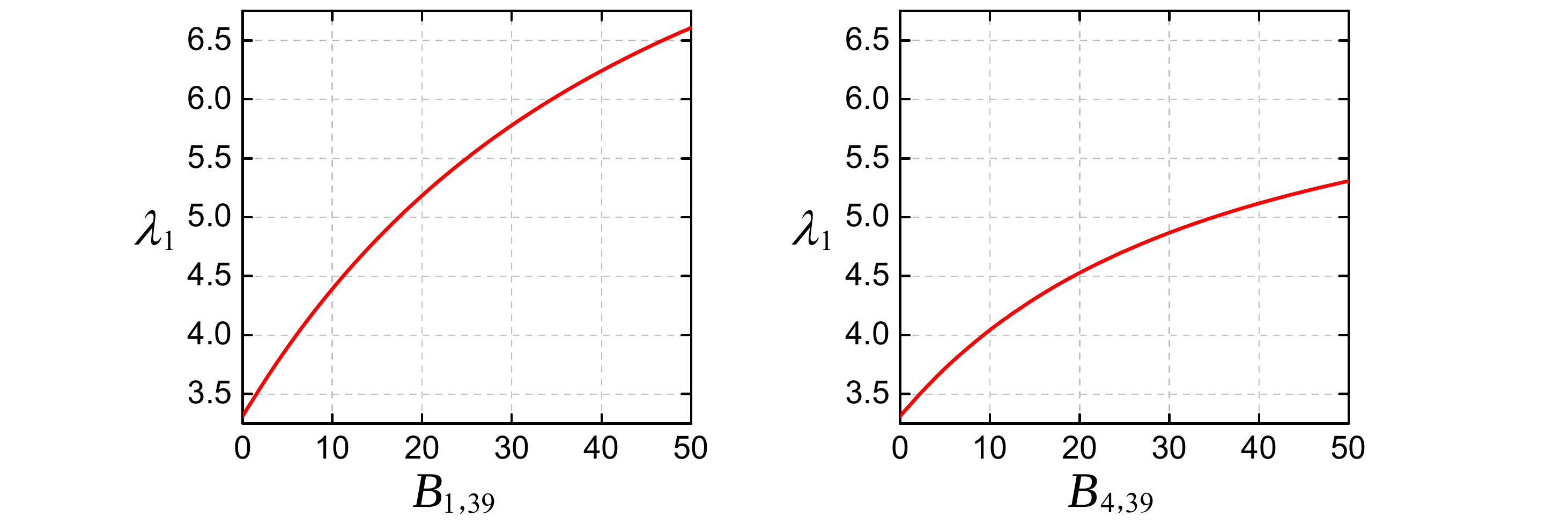}
	\vspace{-3mm}
	%\DeclareGraphicsExtensions.
	\caption{$\lambda _1$ as functions of $B_{1,39}$ and $B_{4,39}$.}
	\label{Fig_Sen_Curves_Self-loops}
\end{figure}

Another convenient approach to improve the stability is to change the placement of the converters \cite{poolla2017optimal,borsche2015effects}. As suggested in the previous subsection, the converter with the largest participation factor (i.e., Converter~9 as shown in Table~\ref{Participation_Factor}) can be moved to the sites that are closer to the converter with the smallest participation factor (i.e., Converter~6 as shown in Table~\ref{Participation_Factor}). Here we choose to move Converter~9 to Node~30, which makes $\lambda_1$ change from 3.3118 to 3.5514 and thus the PLL-synchronization stability is improved.

\begin{figure}[!t]
	\centering
	\includegraphics[width=1.5in]{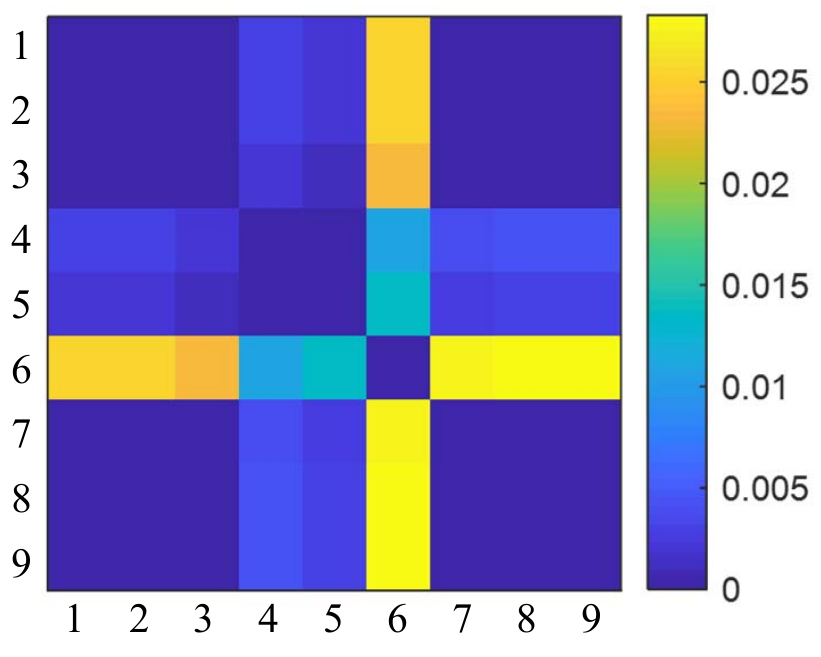}
	\vspace{-3mm}
	%\DeclareGraphicsExtensions.
	\caption{Sensitivity matrix of ${\lambda _1}$ to perturbations among the converter nodes.}
	\label{Fig_Sen_matrix_converter}
\end{figure}

\begin{figure}[!t]
	\centering
	\includegraphics[width=3.5in]{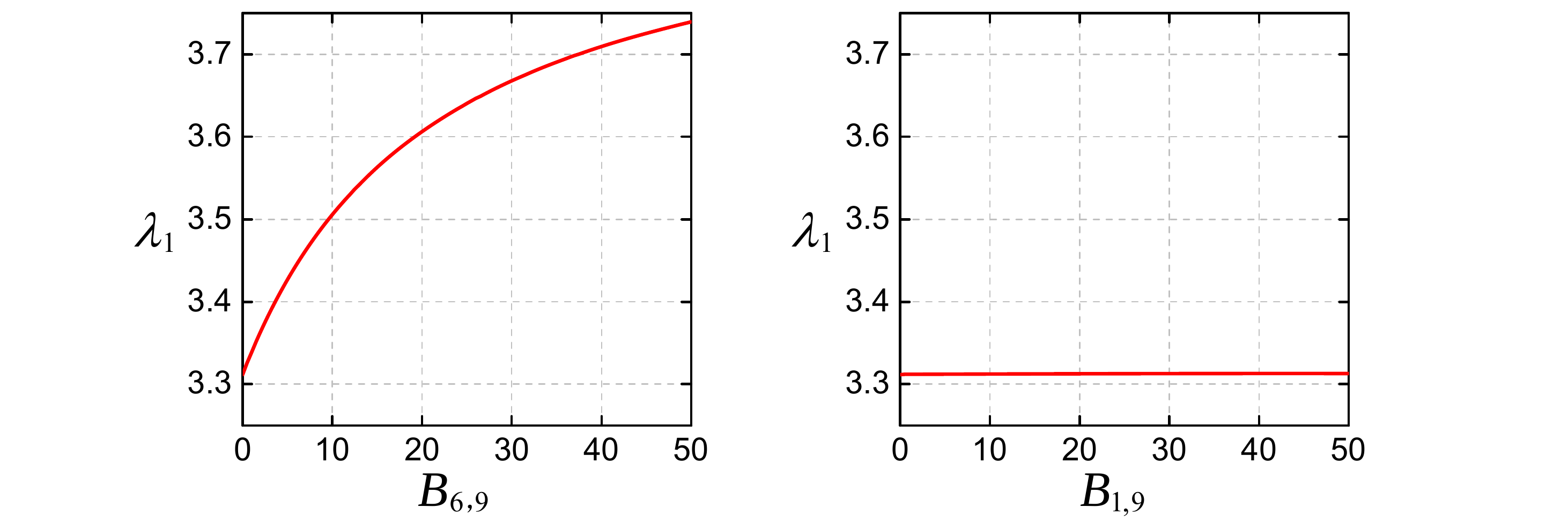}
	\vspace{-5mm}
	%\DeclareGraphicsExtensions.
	\caption{$\lambda _1$ as functions of $B_{6,9}$ and $B_{1,9}$.}
	\vspace{-5mm}
	\label{Fig_Sen_Curves_Converter_nodes}
\end{figure}

To illustrate the effects of perturbations among the converter nodes, Fig.\ref{Fig_Sen_matrix_converter} displays the submatrix of the sensitivity matrix of $\lambda _1$ with $M_{ij} = \partial \lambda _1 / \partial B_{ij}$, $i,j \in \{1,...,m\}$ (calculated by (\ref{eq:lambda_Bij_W'})), and $M_{ij}$ is set to be $0$ if $i=j$. It can be seen that {the most sensitive connection is to node 6, which is the least connected node in Fig.~\ref{Fig_Multi_Converter}. The} maximal entry is $M_{6,9} = 0.0283$, which indicates that adding a new line between node 6 and node 9 has the most significant effect on improving the stability. Besides, adding lines between node 6 and node 7, or between node 6 and node 8 can also effectively improve the stability.

Further, Fig.\ref{Fig_Sen_Curves_Converter_nodes} plots $\lambda _1$ as functions of $B_{6,9}$ and $B_{1,9}$. With the increase of $B_{6,9}$ from 0 to 50, $\lambda _1$ is increased from 3.3118 to 3.7393, which improves the system stability. By comparison, $\lambda _1$ almost remains the same with the increase of $B_{1,9}$ from 0 to 50. The above results are in accordance with Fig.\ref{Fig_Sen_matrix_converter} which shows that $M_{6,9}$ is much greater that $M_{1,9}$.

\section{Simulation Results}

To illustrate the effectiveness of our linearization and sensitivity-based analysis, we now provide a detailed simulation study based on a {\em nonlinear} model of the multi-converter system in Fig.\ref{Fig_Multi_Converter}, with parameters given in Appendix \ref{Appendix: System parameters}.

Fig.\ref{Fig_Case_B3239} displays the time-domain responses of the multi-converter system to show how the changes of grid structure affect the {PLL-synchronization} stability and performance of the system. In Fig.\ref{Fig_Case_B3239} (a), $B_{32,39}$ is decreased from 61.27 to 50 at $t=0.1\rm{s}$, and in Fig.\ref{Fig_Case_B3239} (b), $B_{32,39}$ is decreased from 61.27 to 40 at $t=0.1\rm{s}$. {Due to the nonlinearity of underlying model, the changes result in a transient deviation before relaxing to a (possibly new) equilibrium point.} It can be seen that the active power responses of the nine converters have higher damping ratio in Fig.\ref{Fig_Case_B3239} (a) than those in Fig.\ref{Fig_Case_B3239} (b), consistent with the results in Fig.\ref{Fig_Sen_Curves_IB} that $\lambda _1$ is decreased with the decrease of $B_{32,39}$, thereby deteriorating the system stability. Note that higher damping ratio indicates a larger stability margin. 

Fig.\ref{Fig_Case_B3239} (c) plots the responses when $B_{32,39}$ is decreased from 61.27 to 30 at $t=0.1\rm{s}$, and the system is {linearly} unstable {(resulting in sustained oscillations for the nonlinear system) since} the active power outputs of the nine converters are oscillating and cannot converge. It is consistent with Fig.\ref{Fig_Sen_Curves_IB} that the system will become {linearly} unstable if $B_{32,39}$ is less than 30.95, which causes $\lambda _1 < 2.25$.

\begin{figure}[!t]
	\centering
	\includegraphics[width=2.4in]{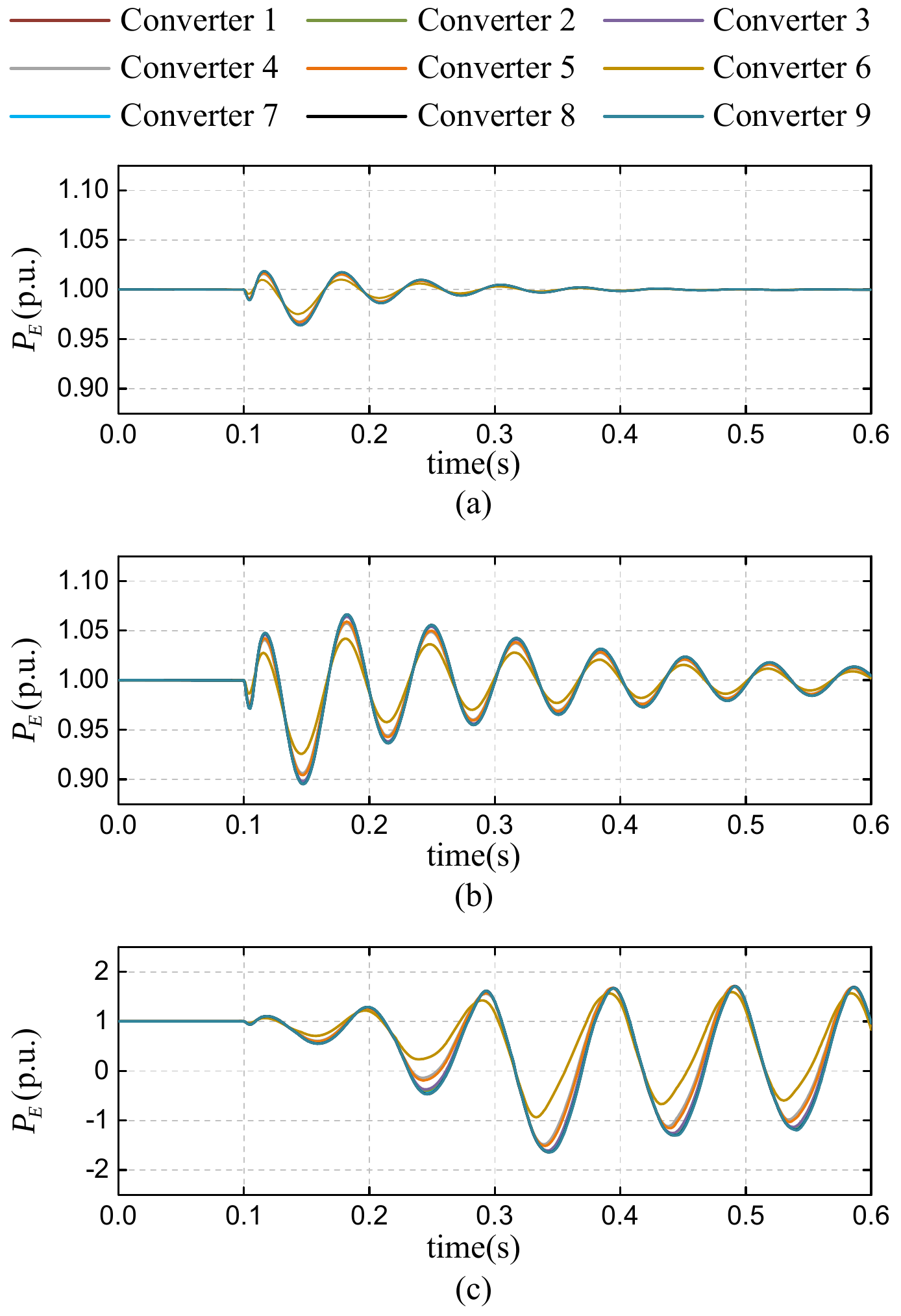}
	\vspace{-3mm}
	%\DeclareGraphicsExtensions.
	\caption{Time-domain responses of the multi-converter system (a) $B_{32,39}$ is decreased from 61.27 to 50 (b)  $B_{32,39}$ is decreased from 61.27 to 40 (c) $B_{32,39}$ is decreased from 61.27 to 30 (an unstable case).}
	\vspace{-3mm}
	\label{Fig_Case_B3239}
\end{figure}

To illustrate how perturbations of the existing lines affect the {PLL-synchronization} stability, Fig.\ref{Fig_Case_B3233} shows the active power responses of the nine converters when $B_{32,33}$ and $B_{17,18}$ are perturbed. In Fig.\ref{Fig_Case_B3233} (a), $B_{32,33}$ is increased from 47.62 to 95.24 at $t=0.1\rm{s}$, and in Fig.\ref{Fig_Case_B3233} (b), $B_{32,33}$ is decreased from 47.62 to 22.955 at $t=0.1\rm{s}$. It can be seen that increasing $B_{32,33}$ can increase the damping ratio and improve the stability. Fig.\ref{Fig_Case_B3233} (c) and (d) plot the responses when $B_{17,18}$ is changed, and the system has little response in these two cases because $\lambda _1$ is not sensitive to the perturbations on $B_{17,18}$, as demonstrated in Fig.\ref{Fig_Sensitivity_Matrix}. These simulation results are consistent with the results in Fig.\ref{Fig_Sen_Curves} as well. 

\begin{figure}[!t]
	\centering
	\includegraphics[width=2.6in]{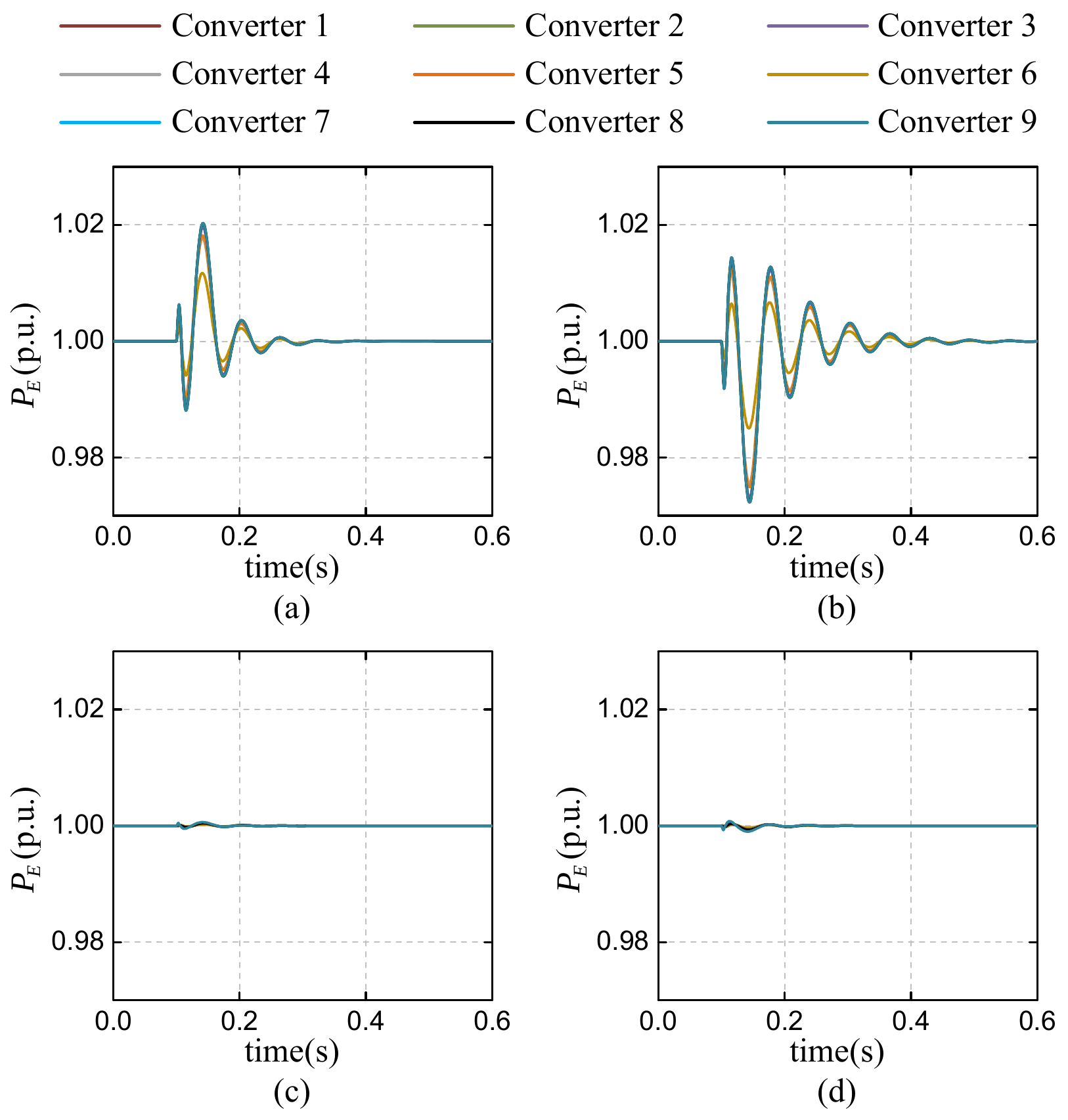}
	\vspace{-3mm}
	%\DeclareGraphicsExtensions.
	\caption{Time-domain responses of the system (a) $B_{32,33}$ is increased from 47.62 to 95.24 (b)  $B_{32,33}$ is decreased from 47.62 to 23.82 (c) $B_{17,18}$ is increased from 45.91 to 91.82 (d) $B_{17,18}$ is decreased from 45.91 to 22.955.}
	\vspace{-3mm}
	\label{Fig_Case_B3233}
\end{figure}

\section{Conclusions}
This paper investigated the impacts of grid structure on PLL-synchronization stability of multi-converter systems. The stability analysis of a single-converter infinite-bus system demonstrated that the stability margin of PLL-based converter is strongly related to the grid-side admittance. We explicitly showed how the dynamics of the converters in a multi-converter system are coupled via the power network, and how they can be decoupled into several modal subsystems when the converters have homogeneous dynamics. On this basis, we revealed that the smallest eigenvalue of the grounded Laplacian matrix of the Kron-reduced power network dominates the PLL-synchronization stability margin of the whole multi-converter system. Moreover, we showed that the PLL bandwidth should be reduced with the decrease of this smallest eigenvalue in order to ensure the stability of multi-converter systems.
Based on these insights, we revealed the effect of the grid structure on the system stability through the sensitivity of the smallest eigenvalue with respect to network perturbations.
Further, we provided guidelines on how to improve the system stability by proper placement of the converters and enhancing some weak connections, which is particularly useful for the grid planning. Our results confirm the prevailing intuition, that PLL-based power converters are stable only in a strong grid.

Future work will focus on how to optimize the grid structure to enhance the PLL-synchronization stability by further considering economic and geographic factors.

\bibliographystyle{IEEEtran}

\bibliography{ref}

% Generated by IEEEtran.bst, version: 1.14 (2015/08/26)
\begin{thebibliography}{10}
\providecommand{\url}[1]{#1}
\csname url@samestyle\endcsname
\providecommand{\newblock}{\relax}
\providecommand{\bibinfo}[2]{#2}
\providecommand{\BIBentrySTDinterwordspacing}{\spaceskip=0pt\relax}
\providecommand{\BIBentryALTinterwordstretchfactor}{4}
\providecommand{\BIBentryALTinterwordspacing}{\spaceskip=\fontdimen2\font plus
\BIBentryALTinterwordstretchfactor\fontdimen3\font minus
  \fontdimen4\font\relax}
\providecommand{\BIBforeignlanguage}[2]{{%
\expandafter\ifx\csname l@#1\endcsname\relax
\typeout{** WARNING: IEEEtran.bst: No hyphenation pattern has been}%
\typeout{** loaded for the language `#1'. Using the pattern for}%
\typeout{** the default language instead.}%
\else
\language=\csname l@#1\endcsname
\fi
#2}}
\providecommand{\BIBdecl}{\relax}
\BIBdecl

\bibitem{carrasco2006power}
J.~M. Carrasco, L.~G. Franquelo, J.~T. Bialasiewicz \emph{et~al.},
  ``Power-electronic systems for the grid integration of renewable energy
  sources: A survey,'' \emph{IEEE Trans. Ind. Electron.}, vol.~53, no.~4, pp.
  1002--1016, 2006.

\bibitem{rocabert2012control}
J.~Rocabert, A.~Luna, F.~Blaabjerg, and P.~Rodriguez, ``Control of power
  converters in ac microgrids,'' \emph{IEEE Trans. Power Electron.}, vol.~27,
  no.~11, pp. 4734--4749, 2012.

\bibitem{FM-FD-GH-DH-GV:18}
F.~Milano, F.~D{\"o}rfler, G.~Hug, D.~Hill, and G.~Verbic, ``Foundations and
  challenges of low-inertia systems,'' in \emph{Power Systems Computation
  Conference (PSCC)}, Dublin, Ireland, 2018, {to appear.}

\bibitem{dorfler2012synchronization}
F.~D{\"o}rfler and F.~Bullo, ``Synchronization and transient stability in power
  networks and nonuniform kuramoto oscillators,'' \emph{SIAM Journal on Control
  and Optimization}, vol.~50, no.~3, pp. 1616--1642, 2012.

\bibitem{blaabjerg2006overview}
F.~Blaabjerg, R.~Teodorescu, M.~Liserre, and A.~V. Timbus, ``Overview of
  control and grid synchronization for distributed power generation systems,''
  \emph{IEEE Trans. Ind. Electron.}, vol.~53, no.~5, pp. 1398--1409, 2006.

\bibitem{golestan2015conventional}
S.~Golestan and J.~M. Guerrero, ``Conventional synchronous reference frame
  phase-locked loop is an adaptive complex filter,'' \emph{IEEE Trans. Ind.
  Electron.}, vol.~62, no.~3, pp. 1679--1682, 2015.

\bibitem{wang2018unified}
X.~Wang, L.~Harnefors, and F.~Blaabjerg, ``Unified impedance model of
  grid-connected voltage-source converters,'' \emph{IEEE Trans. Power
  Electron.}, vol.~33, no.~2, pp. 1775--1787, 2018.

\bibitem{wen2016analysis}
B.~Wen, D.~Boroyevich, R.~Burgos, P.~Mattavelli, and Z.~Shen, ``Analysis of dq
  small-signal impedance of grid-tied inverters,'' \emph{IEEE Trans. Power
  Electron.}, vol.~31, no.~1, pp. 675--687, 2016.

\bibitem{liu2018frequency}
W.~Liu, Z.~Lu, X.~Wang, and X.~Xie, ``Frequency-coupled admittance modelling of
  grid-connected voltage source converters for the stability evaluation of
  subsynchronous interaction,'' \emph{IET Renewable Power Generation}, vol.~13,
  no.~2, pp. 285--295, 2018.

\bibitem{huang2018adaptive}
L.~Huang, H.~Xin, Z.~Wang \emph{et~al.}, ``An adaptive phase-locked loop to
  improve stability of voltage source converters in weak grids,'' in
  \emph{Power and Energy Society General Meeting (PESGM), Portland, OR,
  USA}.\hskip 1em plus 0.5em minus 0.4em\relax IEEE, 2018, pp. 1--5.

\bibitem{purba2017network}
V.~Purba, S.~V. Dhople, S.~Jafarpour, F.~Bullo, and B.~B. Johnson,
  ``Network-cognizant model reduction of grid-tied three-phase inverters,'' in
  \emph{Communication, Control, and Computing (Allerton), 2017 55th Annual
  Allerton Conference on}.\hskip 1em plus 0.5em minus 0.4em\relax IEEE, 2017,
  pp. 157--164.

\bibitem{jafarpour2019small}
S.~Jafarpour, V.~Purba, S.~V. Dhople, B.~Johnson, and F.~Bullo, ``Small-signal
  stability of grid-tied inverter networks,'' \emph{arXiv preprint
  arXiv:1902.02478}, 2019.

\bibitem{xin2017generalized}
H.~Xin, W.~Dong, D.~Gan, D.~Wu, and X.~Yuan, ``Generalized short circuit ratio
  for multi power electronic based devices infeed systems: Definition and
  theoretical analysis,'' \emph{arXiv preprint arXiv:1708.08046}, 2017.

\bibitem{hill2006power}
D.~J. Hill and G.~Chen, ``Power systems as dynamic networks,'' in
  \emph{International Symposium on Circuits and Systems}.\hskip 1em plus 0.5em
  minus 0.4em\relax IEEE, 2006.

\bibitem{dorfler2013synchronization}
F.~D{\"o}rfler, M.~Chertkov, and F.~Bullo, ``Synchronization in complex
  oscillator networks and smart grids,'' \emph{Proceedings of the National
  Academy of Sciences}, vol. 110, no.~6, pp. 2005--2010, 2013.

\bibitem{song2017network}
Y.~Song, D.~Hill, and T.~Liu, ``Network-based analysis of small-disturbance
  angle stability of power systems,'' \emph{IEEE Trans. Control Netw. Syst.},
  2017, {to appear.}

\bibitem{simpson2013synchronization}
J.~W. Simpson-Porco, F.~D{\"o}rfler, and F.~Bullo, ``Synchronization and power
  sharing for droop-controlled inverters in islanded microgrids,''
  \emph{Automatica}, vol.~49, no.~9, pp. 2603--2611, 2013.

\bibitem{wen2015impedance}
B.~Wen, D.~Dong, D.~Boroyevich, R.~Burgos, P.~Mattavelli, and Z.~Shen,
  ``Impedance-based analysis of grid-synchronization stability for three-phase
  paralleled converters,'' \emph{IEEE Transactions on Power Electronics},
  vol.~31, no.~1, pp. 26--38, 2015.

\bibitem{harnefors2007input}
L.~Harnefors, M.~Bongiorno, and S.~Lundberg, ``Input-admittance calculation and
  shaping for controlled voltage-source converters,'' \emph{IEEE transactions
  on industrial electronics}, vol.~54, no.~6, pp. 3323--3334, 2007.

\bibitem{harnefors2007modeling}
L.~Harnefors, ``Modeling of three-phase dynamic systems using complex transfer
  functions and transfer matrices,'' \emph{IEEE Trans. Ind. Electron.},
  vol.~54, no.~4, pp. 2239--2248, 2007.

\bibitem{huang2019grid}
L.~Huang, H.~Xin, Z.~Li, P.~Ju, H.~Yuan, Z.~Lan, and Z.~Wang,
  ``Grid-synchronization stability analysis and loop shaping for pll-based
  power converters with different reactive power control,'' \emph{IEEE
  Transactions on Smart Grid}, 2019.

\bibitem{dorfler2013kron}
F.~D{\"o}rfler and F.~Bullo, ``Kron reduction of graphs with applications to
  electrical networks.'' \emph{IEEE Trans. on Circuits and Systems}, vol.~60,
  no.~1, pp. 150--163, 2013.

\bibitem{motter2013spontaneous}
A.~E. Motter, S.~A. Myers, M.~Anghel, and T.~Nishikawa, ``Spontaneous synchrony
  in power-grid networks,'' \emph{Nature Physics}, vol.~9, no.~3, p. 191, 2013.

\bibitem{jiang2017performance}
Y.~Jiang, R.~Pates, and E.~Mallada, ``Performance tradeoffs of dynamically
  controlled grid-connected inverters in low inertia power systems,'' in
  \emph{2017 IEEE 56th Annual Conference on Decision and Control (CDC)}.\hskip
  1em plus 0.5em minus 0.4em\relax IEEE, 2017, pp. 5098--5105.

\bibitem{dong2018small}
W.~Dong, H.~Xin, D.~Wu, and L.~Huang, ``Small signal stability analysis of
  multi-infeed power electronic systems based on grid strength assessment,''
  \emph{IEEE Trans. Power Syst.}, 2018, {early access.}

\bibitem{skogestad2007multivariable}
S.~Skogestad and I.~Postlethwaite, \emph{Multivariable feedback control:
  analysis and design}.\hskip 1em plus 0.5em minus 0.4em\relax Wiley New York,
  2007, vol.~2.

\bibitem{pirani2014spectral}
M.~Pirani and S.~Sundaram, ``Spectral properties of the grounded laplacian
  matrix with applications to consensus in the presence of stubborn agents,''
  in \emph{American Control Conference (ACC), 2014}.\hskip 1em plus 0.5em minus
  0.4em\relax IEEE, 2014, pp. 2160--2165.

\bibitem{pirani2016smallest}
------, ``On the smallest eigenvalue of grounded laplacian matrices,''
  \emph{IEEE Transactions on Automatic Control}, vol.~61, no.~2, pp. 509--514,
  2016.

\bibitem{dorfler2018electrical}
F.~D{\"o}rfler, J.~W. Simpson-Porco, and F.~Bullo, ``Electrical networks and
  algebraic graph theory: Models, properties, and applications,''
  \emph{Proceedings of the IEEE}, vol. 106, no.~5, pp. 977--1005, 2018.

\bibitem{giordano2016smallest}
G.~Giordano, F.~Blanchini, E.~Franco, V.~Mardanlou, and P.~L. Montessoro, ``The
  smallest eigenvalue of the generalized laplacian matrix, with application to
  network-decentralized estimation for homogeneous systems,'' \emph{IEEE
  Transactions on Network Science and Engineering}, vol.~3, no.~4, pp.
  312--324, 2016.

\bibitem{zhang2017generalized}
F.~Zhang, H.~Xin, Z.~Wang, D.~Gan, Q.~Xu, P.~Dai, and F.~Liu, ``Generalized
  short circuit ratio for multi-infeed lcc-hvdc systems,'' in \emph{Power and
  Energy Society General Meeting (PESGM)}.\hskip 1em plus 0.5em minus
  0.4em\relax IEEE, 2017, pp. 1--5.

\bibitem{huang2018damping}
L.~Huang, H.~Xin, and Z.~Wang, ``Damping low-frequency oscillations through
  vsc-hvdc stations operated as virtual synchronous machines,'' \emph{IEEE
  Trans. Power Electron.}, 2018, {early access.}

\bibitem{ma2006eigenvalue}
J.~Ma, Z.~Y. Dong, and P.~Zhang, ``Eigenvalue sensitivity analysis for dynamic
  power system,'' in \emph{International Conference on Power System
  Technology}.\hskip 1em plus 0.5em minus 0.4em\relax IEEE, 2006, pp. 1--7.

\bibitem{kundur1994power}
P.~Kundur, N.~J. Balu, and M.~G. Lauby, \emph{Power system stability and
  control}.\hskip 1em plus 0.5em minus 0.4em\relax McGraw-hill New York, 1994,
  vol.~7.

\bibitem{poolla2017optimal}
B.~K. Poolla, S.~Bolognani, and F.~D{\"o}rfler, ``Optimal placement of virtual
  inertia in power grids,'' \emph{IEEE Transactions on Automatic Control},
  vol.~62, no.~12, pp. 6209--6220, 2017.

\bibitem{borsche2015effects}
T.~S. Borsche, T.~Liu, and D.~J. Hill, ``Effects of rotational inertia on power
  system damping and frequency transients,'' in \emph{2015 54th IEEE conference
  on decision and control (CDC)}.\hskip 1em plus 0.5em minus 0.4em\relax IEEE,
  2015, pp. 5940--5946.

\end{thebibliography}

%\newpage
%
%\clearpage

\appendices
\vspace{-2mm}
\section{System parameters}
\label{Appendix: System parameters}

\renewcommand{\thetable}{\thesection.\arabic{table}}

\setcounter{table}{0}

\vspace{-4mm}

\renewcommand\arraystretch{1.4}
\begin{table}[h]
	\scriptsize
	\centering
	\caption{Parameters of the Multi-Converter System}
	\begin{tabular}{|lll|}
		\hline
		\multicolumn{3}{|c|}{Base Values for Per-unit Calculation}										\\
		\hline
		$U_{\rm{base}} = 380\rm{V}$	&	$S_{\rm{base}} = 50\rm{kVA}$	&	$f_{\rm{base}} = 50\rm{Hz}$	\\
		\hline
		\multicolumn{3}{|c|}{Parameters of the Power Converters (p.u.)}									\\
		\hline
		$L_F = 0.05$				&	$C_F = 0.05$					&	$L_g = 0.05$				\\
		$K_{\rm VF} = 1$				&	$T_{\rm VF} = 0.01$					&	$K_{\rm CCP} = 0.3$				\\
		$K_{\rm CCI} = 10$				&	$K_{\rm PCP} = 0.5$					&	$K_{\rm PCI} = 40$				\\
		$K_{\rm QCP} = 0.5$				&	$K_{\rm QCI} = 40$					&	$K_{\rm PLLP} = 34.36$			\\
		$K_{\rm PLLI} = 590.17$			&	$P^{\rm ref} = 1$					&	$Q^{\rm ref} = 0$				\\
		\hline
		\multicolumn{3}{|c|}{Susceptance of the Electrical Network (p.u.)}								\\
		\hline
		$B_{1,11} = 92.08$			&	$B_{2,15} = 66.67$				&	$B_{3,19} = 83.33$			\\
		$B_{4,28} = 117.37$			&	$B_{5,29} = 92.59$				&	$B_{6,31} = 116.55$			\\
		$B_{7,34} = 71.84$			&	$B_{8,38} = 106.84$				&	$B_{9,18} = 66.67$			\\
		$B_{10,11} = 40.55$			&	$B_{11,12} = 110.38$			&	$B_{12,13} = 78.25$			\\
		$B_{13,14} = 130.21$		&	$B_{14,15} = 641.03$			&	$B_{15,16} = 181.16$		\\
		$B_{16,17} = 362.32$		&	$B_{17,18} = 45.91$				&	$B_{15,20} = 203.25$		\\
		$B_{19,20} = 387.60$		&	$B_{20,21} = 38.31$				&	$B_{19,22} = 387.60$		\\
		$B_{21,22} = 38.31$			&	$B_{13,23} = 129.20$			&	$B_{22,23} = 165.02$		\\
		$B_{23,24} = 76.80$			&	$B_{24,25} = 177.31$			&	$B_{25,26} = 187.27$		\\
		$B_{12,27} = 125.31$		&	$B_{26,27} = 203.25$			&	$B_{25,28} = 85.47$			\\
		$B_{28,29} = 120.77$		&	$B_{25,30} = 123.46$			&	$B_{30,31} = 119.05$		\\
		$B_{31,32} = 173.61$		&	$B_{25,33} = 282.49$			&	$B_{32,33} = 47.62$			\\
		$B_{11,34} = 193.80$		&	$B_{34,35} = 52.60$				&	$B_{26,36} = 96.34$			\\
		$B_{35,36} = 113.38$		&	$B_{35,37} = 35.16$				&	$B_{35,38} = 26.67$			\\
		$B_{37,38} = 110.38$		&	$B_{32,39} = 61.27$				&	$B_{9,10} = 66.67$			\\
		$B_{14,17} = 148.81$		&									&								\\
		\hline
	\end{tabular}	
	\vspace{-4mm}	
	\label{network_parameter}
\end{table}

\section{admittance Matrix of PLL-based Converter}
\label{Appenxix:admittance Matrix}

\renewcommand{\theequation}{\thesection.\arabic{equation}}

\setcounter{equation}{0}
In the following, we present the derivation of the admittance matrix of the PLL-based converter in Fig.\ref{Fig_Single_Converter}.
When modeled in the controller's rotating \emph{dq}-frame (whose angular frequency is determined by the PLL), the dynamic equations of the \emph{LCL} filter can be expressed as \cite{harnefors2007modeling}
\begin{equation}
{\vec U^*} - \vec V = \left( {s{L_F} + j\omega {L_F}} \right) \times \vec I_C\,,		\label{eq:Lf}
\end{equation}
\begin{equation}
\vec I_C - {\vec I} = \left( {s{C_F} + j\omega {C_F}} \right) \times \vec V  \buildrel \Delta \over = {\vec Y_{\rm CL}}(s) \times \vec V\,, 	\label{eq:Cf}
\end{equation}
\begin{equation}
\vec V - \vec U = \left( {s{L_g} + j\omega {L_g}} \right) \times {\vec I}  \buildrel \Delta \over = {\vec Z_g}(s) \times {\vec I} 	\label{eq:Lg}{\,,}
\end{equation}
where ${\vec U^*} = U_d^* + jU_q^*$, $\vec V = {V_d} + j{V_q}$, $\vec U = {U_d} + j{U_q}$, $\vec I_C = {I_{Cd}} + j{I_{Cq}}$ and ${\vec I} = {I_{d}} + j{I_{q}}$ are the corresponding vectors of ${\bf{U}}_{{\bf{abc}}}^{\bf{*}}$, ${{\bf{V}}_{{\bf{abc}}}}$, ${{\bf{U}}_{{\bf{abc}}}}$, ${{\bf{I}}_{{\bf{Cabc}}}}$ and ${{\bf{I}}_{{\bf{abc}}}}$ in the controller's \emph{dq}-frame, respectively. ${L_F}$ is the converter-side inductance, ${L_g}$ is the grid-side inductance, and ${C_F}$ is the \emph{LCL}'s capacitance. 

The dynamic equation of the current control loop is
\begin{equation}
{\vec U^*} = PI_{\rm CC}(s) \times \left( {{{\vec I}^{ref}} - \vec I_C} \right ) + j\omega {L_F}\vec I_C + {f_{\rm VF}}(s)\vec V\,,	\label{eq:current_control}
\end{equation}
where $P{I_{\rm CC}}(s) = {K_{\rm CCP}} + {{{K_{\rm CCI}}} \mathord{\left/ {\vphantom {{{K_{\rm CCI}}} s}} \right. \kern-\nulldelimiterspace} s}$ is the transfer function of the PI regulator, ${f_{\rm VF}}(s) = {{{K_{\rm VF}}} \mathord{\left/{\vphantom {{{K_{\rm VF}}} {\left( {{T_{\rm VF}}s + 1} \right)}}} \right.\kern-\nulldelimiterspace} {\left( {{T_{\rm VF}}s + 1} \right)}}$ is a first-order filter that mitigates the high-frequency components of the voltage feed-forward signals, and ${\vec I^{\rm ref}} = I_d^{\rm ref} + jI_q^{\rm ref}$ is the current reference vector that comes from the power control.

Substituting (\ref{eq:Lf}) into (\ref{eq:current_control}) yields
\begin{equation}
{G_I}(s) \times {\vec I^{\rm ref}} - {Y_{\rm VF}}(s) \times \vec V = \vec I_C\,,		\label{eq:current_track}
\end{equation}
where
\begin{equation}
{G_I}(s) = \frac{{P{I_{\rm CC}}(s)}}{{s{L_F} + P{I_{\rm CC}}(s)}},\;{Y_{\rm VF}}(s) = \frac{{1 - {f_{\rm VF}}(s)}}{{s{L_F} + P{I_{\rm CC}}(s)}}.	\label{eq:GI}
\end{equation}

We note that the above equations are obtained based on space vectors and complex transfer functions, and they can be conveniently transformed to matrix form considering the following equivalent transformation \cite{harnefors2007modeling}
\begin{equation}
\begin{split}
{y_d} + j{y_q} &= \left[ {{G_d}(s) + j{G_q}(s)} \right] \times \left( {{x_d} + j{x_q}} \right)\\
\Leftrightarrow \left[ {\begin{array}{*{20}{c}}
	{{y_d}}\\
	{{y_q}}
	\end{array}} \right] &= \left[ {\begin{array}{*{20}{c}}
	{{G_d}(s)}&{ - {G_q}(s)}\\
	{{G_q}(s)}&{{G_d}(s)}
	\end{array}} \right]\left[ {\begin{array}{*{20}{c}}
	{{x_d}}\\
	{{x_q}}
	\end{array}} \right].	\label{eq:transf}
\end{split}
\end{equation}

The control law of the SRF-PLL that determines the dynamics of the controller's rotating \emph{dq}-frame is
\begin{equation}
\theta  = \frac{\omega }{s} = \frac{1}{s} \times \left( {{K_{\rm PLLP}} + \frac{{{K_{\rm PLLI}}}}{s}} \right) \times {V_q} \buildrel \Delta \over = {f_{\rm PLL}}(s) \times {V_q}\,,		\label{eq:PLL}
\end{equation}
where $\theta$ (rad) is the phase of the controller's rotating \emph{dq}-frame and $\omega$ (rad/s) is the angular frequency. $K_{\rm PLLP}$ and $K_{\rm PLLI}$ are the parameters of the PI regulator.

The converter applies active power control and reactive power control which can be formulated by
\begin{equation}
\begin{split}
I_d^{\rm ref} &= P{I_{\rm PC}}(s) \times \left( {{P^{\rm ref}} - {P_E}} \right)\,,\\
I_q^{\rm ref} &= P{I_{\rm QC}}(s) \times \left( {{Q_E} - {Q^{\rm ref}}} \right)\,,		\label{eq:power_control}
\end{split}
\end{equation}
where $P{I_{\rm PC}}(s) = {K_{\rm PCP}} + {{{K_{\rm PCI}}} \mathord{\left/{\vphantom {{{K_{\rm PCI}}} s}} \right.\kern-\nulldelimiterspace} s}$ and $P{I_{\rm QC}}(s) = {K_{\rm QCP}} + {{{K_{\rm QCI}}} \mathord{\left/{\vphantom {{{K_{QCI}}} s}} \right.\kern-\nulldelimiterspace} s}$ are the transfer functions of the PI regulators, $P^{\rm ref}$ and $Q^{\rm ref}$ are the reference values, $P_E$ and $Q_E$ are the active and reactive power of the converter (see Fig.\ref{Fig_Single_Converter}) which can be calculated by
\begin{equation}
\begin{split}
{P_E} &= {V_d}{I_{Cd}} + {V_q}{I_{Cq}}\,,\\
{Q_E} &= {V_q}{I_{Cd}} - {V_d}{I_{Cq}}.		\label{eq:PQ}
\end{split}
\end{equation}
By linearizing (\ref{eq:PQ}) {around the equilibrium point} $(I_{Cd0},I_{Cq0},V_{d0},V_{q0})$ and combining it with (\ref{eq:current_track}) and (\ref{eq:power_control}) yields the converter-side equivalent admittance
\begin{equation}
- \left[ {\begin{array}{*{20}{c}}
	{\Delta {I_{Cd}}}\\
	{\Delta {I_{Cq}}}
	\end{array}} \right] = \left[ {\begin{array}{*{20}{c}}
	{{Y_{11}}(s)}&{{Y_{12}}(s)}\\
	{{Y_{21}}(s)}&{{Y_{22}}(s)}
	\end{array}} \right]\left[ {\begin{array}{*{20}{c}}
	{\Delta {V_d}}\\
	{\Delta {V_q}}
	\end{array}} \right]\,,		\label{eq:I_V_matrix}
\end{equation}
where
\begin{equation}
\begin{split}
{Y_{11}}(s) &= \frac{{{G_I}(s)P{I_{\rm PC}}(s){I_{Cd0}} + {Y_{\rm VF}}(s)}}{{{\rm{1 + }}{G_I}(s)P{I_{\rm PC}}(s){V_{d0}}}}\,,\\
{Y_{12}}(s) &= \frac{{{G_I}(s)P{I_{\rm PC}}(s){I_{Cq0}}}}{{{\rm{1 + }}{G_I}(s)P{I_{\rm PC}}(s){V_{d0}}}}\,,\\
{Y_{21}}(s) &= \frac{{{G_I}(s)P{I_{\rm QC}}(s){I_{Cq0}}}}{{1 + {G_I}(s)P{I_{\rm QC}}(s){V_{d0}}}}\,,\\
{Y_{22}}(s) &= \frac{{ - {G_I}(s)P{I_{\rm QC}}(s){I_{Cd0}} + {Y_{\rm VF}}(s)}}{{1 + {G_I}(s)P{I_{\rm QC}}(s){V_{d0}}}}.
\end{split}
\end{equation}

The equivalent admittance in (\ref{eq:I_V_matrix}) represent the converter-side dynamics in the controller's \emph{dq}-frame, and to derive the closed-loop dynamics of a multi-converter system, the equivalent admittances of the converters need to be in one common coordinate. In this paper, we choose the infinite bus's rotating \emph{dq}-frame as the global coordinate, whose angular frequency is a constant (i.e., $\omega_0 = 100\pi $ rad/s in this paper).

Consider the following coordinate transformation
\begin{equation}
\vec V \times {e^{j\theta }} = \vec V' \times {e^{j{\theta _G}}}\,,		\label{eq:V_transf}
\end{equation}
\begin{equation}
\vec I_C \times {e^{j\theta }} = \vec I'_C \times {e^{j{\theta _G}}}\,,		\label{eq:I_transf}
\end{equation}
where $\vec V' = {V'_d} + j{V'_q}$ and $\vec I'_C = {I'_{Cd}} + j{I'_{Cq}}$ are the corresponding voltage and current vectors in the global coordinate, $\theta _G$ is the phase of the global coordinate which meets $s\theta _G = \omega_0$. 

Linearizing (\ref{eq:PLL}), (\ref{eq:V_transf}) and (\ref{eq:I_transf}) {around the equilibrium point}  $(I_{Cd0},I_{Cq0},V_{d0},V_{q0},\theta_{0},\theta_{G0})$ and then substituting them into (\ref{eq:I_V_matrix}) yields the equivalent admittance as
\begin{equation}
\begin{split}
- \left[ {\begin{array}{*{20}{c}}
	{\Delta {I'_{Cd}}}\\
	{\Delta {I'_{Cq}}}
	\end{array}} \right] &= {\bf{Y'}}(s)\left[ {\begin{array}{*{20}{c}}
	{\Delta {V'_d}}\\
	{\Delta {V'_q}}
	\end{array}} \right]\,,\\
{\bf{Y'}}(s) &= {e^{J{\delta _0}}}\left[ {\begin{array}{*{20}{c}}
	{{Y_{11}}(s)}&{\frac{{{Y_{12}}(s) + {f_{\rm PLL}}(s) \times {I_{q0}}}}{{1 + {f_{\rm PLL}}(s) \times {V_{d0}}}}}\\
	{{Y_{21}}(s)}&{\frac{{{Y_{22}}(s) - {f_{\rm PLL}}(s) \times {I_{d0}}}}{{1 + {f_{\rm PLL}}(s) \times {V_{d0}}}}}
	\end{array}} \right]{e^{ - J{\delta _0}}}\,,		\label{eq:Y'}
\end{split}
\end{equation}
where ${\delta _0} = {\theta _0} - {\theta _{G0}}$, ${e^{J{\delta _0}}}$ is the matrix form of ${e^{j{\delta _0}}}$, and ${e^{-J{\delta _0}}}$ is the matrix form of ${e^{-j{\delta _0}}}$. {Note the asymmetry of the transfer admittance ${\bf{Y'}}(s)$ matrix due to the fact that the PLL (\ref{eq:PLL}) uses only the $q$-component of the voltage.}

Then, considering that ${\vec Y_C}(s)$ and ${\vec Z_g}(s)$ {(defined in (\ref{eq:Cf}) and (\ref{eq:Lg}))} remain the same when formulated in the global coordinate, by combining (\ref{eq:Cf}), (\ref{eq:Lg}) and (\ref{eq:Y'}) one obtains
\begin{equation}
\begin{split}
- \left[ {\begin{array}{*{20}{c}}
	{\Delta {I'_{d}}}\\
	{\Delta {I'_{q}}}
	\end{array}} \right] &= {{\bf{Y}}_{{\bf{C}}}}(s)\left[ {\begin{array}{*{20}{c}}
	{\Delta {U'_d}}\\
	{\Delta {U'_q}}
	\end{array}} \right]\,,\\
{{\bf{Y}}_{{\bf{C}}}}(s) &= {\left\{ {{{\left[ {{{\bf{Y}}_{\bf{CL}}}(s) + {\bf{Y'}}(s)} \right]}^{ - 1}} + {{\bf{Z}}_{\bf{g}}}(s)} \right\}^{ - 1}}\,,
\end{split}		\label{eq:Y_G}
\end{equation}
where ${\left[ {\begin{array}{*{20}{c}} {\Delta {I'_{d}}}&{\Delta {I'_{q}}} \end{array}} \right]^\top}$ and ${\left[ {\begin{array}{*{20}{c}} {\Delta {U'_d}}&{\Delta {U'_q}} \end{array}} \right]^\top}$ are the perturbed vectors of the converter's current output and terminal voltage in the global $dq$-frame, ${{{\bf{Y}}_{\bf{CL}}}(s)}$ and ${{{\bf{Z}}_{\bf{g}}}(s)}$ are the matrix forms of ${\vec Y_{\rm CL}}(s)$ and ${\vec Z_g}(s)$ {via transformation (\ref{eq:transf})}, respectively. It is worth mentioning that the model in (\ref{eq:Y_G}) {is an extension} of the admittance model developed in \cite{wang2018unified} by further taking into account the dynamics of power control loops and voltage feedforward control.

\end{document}